\documentclass{article}

\usepackage[final, nonatbib]{neurips_2022}

\usepackage[utf8]{inputenc} % allow utf-8 input
\usepackage[T1]{fontenc}    % use 8-bit T1 fonts
\usepackage{hyperref}       % hyperlinks
\usepackage{url}            % simple URL typesetting
\usepackage{booktabs}       % professional-quality tables
\usepackage{amsfonts}       % blackboard math symbols
\usepackage{nicefrac}       % compact symbols for 1/2, etc.
\usepackage{microtype}      % microtypography
\usepackage{xcolor}         % colors

\usepackage{amsmath}
\usepackage{amssymb}
\usepackage{mathtools}
\usepackage{amsthm}
\usepackage{algorithm}
\usepackage{algorithmic}
\usepackage{graphicx}

\usepackage{enumitem}
\usepackage{subcaption}
\DeclareMathOperator*{\E}{\mathbb{E}}
\let\Pr\relax
\DeclareMathOperator*{\Pr}{\mathbb{P}}

\DeclareMathOperator*{\argmin}{arg\,min}

\theoremstyle{plain}
\newtheorem{theorem}{Theorem}[section]

\newtheorem{lemma}[theorem]{Lemma}
\newtheorem{corollary}[theorem]{Corollary}
\theoremstyle{definition}

\theoremstyle{remark}

\newcommand{\eps}{\varepsilon}

\newcommand{\N}{\mathbb{N}}

\newcommand{\LIS}{\text{LIS}}
\newcommand{\GraphModel}{CLV-B}

\title{(Optimal) Online Bipartite Matching with Degree Information}

\author{
Anders Aamand \\
MIT \\
\texttt{aamand@mit.edu}
\And
Justin Y.\ Chen \\
MIT \\
\texttt{justc@mit.edu}
\And
Piotr Indyk \\
MIT \\
\texttt{indyk@mit.edu}
}

\begin{document}

\maketitle

\begin{abstract}
We propose a model for online graph problems where algorithms are given access to an oracle that predicts (e.g., based on modeling assumptions or on past data) the degrees of nodes in the graph. Within this model, we study the classic problem of online bipartite matching, and a natural greedy matching algorithm called MinPredictedDegree, which uses predictions of the degrees of offline nodes. For the bipartite version of a stochastic graph model due to Chung, Lu, and Vu where the expected values of the offline degrees are known and used as predictions,  we show that MinPredictedDegree stochastically dominates {\em any} other online algorithm, i.e., it is optimal for graphs drawn from this model. 
Since the ``symmetric'' version of the model, where all online nodes are identical, is a special case of the well-studied ``known i.i.d.\ model'', it follows that the competitive ratio of MinPredictedDegree on such inputs is at least 0.7299. 
For the special case of graphs with power law degree distributions, we show that MinPredictedDegree frequently produces matchings almost as large as the true maximum matching on such graphs.
We complement these results with an extensive empirical evaluation showing that MinPredictedDegree compares favorably to state-of-the-art online algorithms for online matching. 

%plaintext version (for submission)
\iffalse
We propose a model for online graph problems where algorithms are given access to an oracle that predicts (e.g., based on past data) the degrees of nodes in the graph. Within this model, we study the classic problem of online bipartite matching, and a natural greedy matching algorithm called MinPredictedDegree, which uses predictions of the degrees of offline nodes. For the bipartite version of a stochastic graph model due to Chung, Lu, and Vu where the expected values of the offline degrees are known and used as predictions, we show that MinPredictedDegree stochastically dominates any other online algorithm, i.e., it is optimal for graphs drawn from this model. Since the "symmetric" version of the model, where all online nodes are identical, is a special case of the well-studied "known i.i.d. model", it follows that the competitive ratio of MinPredictedDegree on such inputs is at least 0.7299. For the special case of graphs with power law degree distributions, we show that MinPredictedDegree frequently produces matchings almost as large as the true maximum matching on such graphs. We complement these results with an extensive empirical evaluation showing that MinPredictedDegree compares favorably to state-of-the-art online algorithms for online matching.  
\fi
\end{abstract}

\section{Introduction}
\label{sec-intro}
Online algorithms are algorithms that process their inputs “on the fly”, making irrevocable decisions based only on the data seen so far.  Since they do not make any assumptions about the future, they are versatile and work even for adversarial inputs. Unfortunately, by focusing on the worst case, their performance in “typical” cases can be sub-optimal. As a result there has been a large body of research studying various relaxations of the worst-case model, where some extra information about the inputs, or the distribution they are selected from, is available~\cite{Uncertainty}.  

Motivated by the developments in machine learning, over the last few years, many papers have studied online algorithms with predictions~\cite{mitzenmacher2020algorithms}. Such algorithms are equipped with a predictor that, when invoked, provides an (imperfect) prediction of some features of the future part of the input, which is then used by the algorithm to improve its performance. The specific information provided by such predictors is problem-dependent. For graph problems studied in this paper, predictions could include:  the list of edges incident to a given vertex~\cite{kumar2019semi}, the weight of an edge adjacent to a given node in an optimal solution~\cite{antoniadis2020secretary}, or vertex weights that guide a proportional allocation scheme~\cite{lavastida2020learnable}.

In this paper we focus on online graph problems, and propose a model where an algorithm is equipped with a “degree predictor”, i.e., an oracle that, given any vertex, predicts the degree of that vertex in the full graph (containing yet-unseen edges).  This predictor has multiple appealing features. First it is simple, natural, and easy to interpret.  Second, it is useful: vertex degree information is employed in many heuristic and approximation algorithms for graph optimization, for problems such as maximum independent set~\cite{halldorsson1997greed} or maximum matching~~\cite{tinhofer1984probabilistic}.  Third (as demonstrated in Section~\ref{sec-experiments}) such predictors can be easily obtained.  Finally, degree prediction is closely related to the problem of estimating the frequencies of elements in a data set\footnote{The degree of a node is simply the number of times the node appears in the union of all edges.}, and frequency predictors have been already shown to improve the performance of algorithms for multiple data analysis problems~\cite{hsu2019learning,jiang2020learning,eden2021learning,du2021putting}.   
% emphasize that this predictor could be easily learned

The specific graph problem studied in this paper is {\em online bipartite matching}, where we are given a bipartite graph $G = (U \cup V, E)$, and the goal is to find a maximum matching in $G$. In the online setting, the set $U$  is known beforehand, while the vertices in $V$ arrive online one by one. When a new vertex $v$ arrives, the edges in $G$ adjacent to $v$ are provided as well. Online maximum bipartite matching is a classic question studied in the online algorithms literature, with many applications~\cite{mehta2013online}. It is known that a randomized online greedy algorithm, called Ranking, computes a matching of size at least $1-1/e \approx 0.632$ times the optimum~\cite{karp1990optimal},  and that this bound is tight in the worst-case. A large body of work studied various relaxations of the problem, obtained by assuming that  vertex arrivals are random~\cite{goel2008online} or that the graph itself is randomly generated from a given ``known i.i.d. model''~\cite{feldman2009online}. In this paper we extend the basic online model by assuming access to a predictor that, given any “offline” vertex $u \in U$, returns an estimate of its degree. (Note that the degree of any vertex in $V$ is known immediately upon its arrival.)

\paragraph{Our results}  We study the following simple greedy algorithm for bipartite matching: upon the arrival of a vertex $v$, if the set of neighbors $N(v)$ of $v$ in $G$ contains any yet-unmatched vertex, the algorithm selects $u \in N(v)$ of minimum predicted degree in $G$ and adds the edge $(u,v)$ to the matching. This algorithm, which we call {\em MinPredictedDegree (MPD)}, is essentially identical\footnote{The main differences are syntactic: the algorithm of \cite{borodin2019conceptually} computes the degrees based on the given ``type graph'' (see Section~\ref{sec-prelim}), while in this paper we allow arbitrary predictors.}  to the algorithm proposed in~\cite{borodin2019conceptually} which in turn was inspired by the offline matching algorithm called MinGreedy~\cite{ tinhofer1984probabilistic}. The intuition is that vertices with higher degree will have more chances to be matched in the future.

Our main contributions are as follows. First, following in a long line of work on the average-case analysis of matching algorithms initiated by~\cite{karp1981maximum}, we analyze MPD under a natural random bipartite graph model we refer to as {\em \GraphModel{}}, a bipartite version of the Chung-Lu-Vu random graph model~\cite{chung2004spectra}.
A \GraphModel{} random graph is parameterized by $n=|U|$, $m=|V|$, and two weight vectors $\mathbf{p} = \{p_i\}_{i=1}^n\in[0,1]^n$ and $\mathbf{q} = \{q_i\}_{i=1}^m\in[0,1]^m$. For any $u_i \in U$ and $v_j \in V$, the edge $\{u_i,v_j\}$ appears in the graph with probability $p_iq_j$ and these events are mutually independent. 
This model corresponds to the setting where consumers pick their edges with probabilities proportional to the vector $\mathbf{p}$ which describes the relative distribution over the producers. 
%\Anders{I no longer think that this is true when we introduce $q$}. 

Many natural families of random graphs can be described in the \GraphModel{} model. Of particular interest is the case when $q=(1,\dots,1)$, corresponding to the consumers picking their edges i.i.d.; we will refer to this case as the {\em symmetric \GraphModel{} model}. The symmetric version can be viewed as a special case of the well-studied known i.i.d. model of \cite{feldman2009online}. If we further let $\mathbf{p}=(p,\dots,p)$, then the \GraphModel{} graph is an Erd\H{o}s-Rényi random bipartite graph with edge probability $p$.

\paragraph{Theoretical Results}
For the \GraphModel{} model and the MPD algorithm which uses the expected degrees as predictions, we make the following theoretical contributions:

\begin{itemize}
\item We show that MPD stochastically dominates {\em any} other online algorithm, i.e., it is optimal for graphs drawn from the \GraphModel{} model (Section~\ref{sec-opt}). Specifically, we show that for any degree distribution, any algorithm $A$ and any integer $t$, the probability that $A$ produces matching of size at least $t$ is upper bounded by the analogous probability for MPD. Since symmetric \GraphModel{} is a special case of the known i.i.d. model, it follows that the competitive ratio of MPD for this model is at least equal to the best competitive ratio of any algorithm that works for the known i.i.d. model. By the result of~\cite{brubach2016new}, this ratio is at least 0.7299. 
%\Anders{But we are no longer in the known i.i.d. model. Maybe we should just specify that we get this bound when $q=(1,\dots,1)$.}
\item We analyze MPD on symmetric \GraphModel{} model with power law degree distribution (Section~\ref{sec-analysis}). Our theoretical predictions demonstrate that the competitive ratio achieved by our algorithm on such graphs is very high. In particular, for several different power law distributions, it exceeds $0.99$.
\item We also analyze MPD on Erd\H{o}s-Rényi bipartite random graphs where all edges appear with the same probability (Appendix~\ref{appendix-uniform}). In particular, the competitive ratio of the algorithm on such graphs is at least $0.831$. Since in this case all expected degrees are equal, the prediction oracle is of no help. Thus, we conjecture that this is the worst distribution for MPD among all distributions in the  \GraphModel{} model class.  
\item Finally, we observe that  the competitive ratio of MPD is $1/2$ for  {\em worst case} graphs, and that this bound is tight. In addition, we show that the worst-case competitive ratio of \emph{any} algorithm with access to the offline degrees is at most $1-1/e$, implying that degree predictions do not help in the worst-case though they prove to be useful in the random model as well as in practice. See Appendix~\ref{appendix-worst-case-bound} for details.
\end{itemize}

\paragraph{Experiments}
We complement our theoretical studies with an extensive empirical evaluation of MPD for multiple random graph models and real graph benchmarks in Section~\ref{sec-experiments}. Our experiments show that, on most benchmarks, MPD has the best performance among about a dozen state-of-the-art online algorithms, even when compared to algorithms that use much more information about the input.
These experimental results demonstrate that MPD performs well beyond the average-case instances we study theoretically.
%Furthermore, we observe that the performance of MPD remains strong even when the prediction  quality degrades (up to a point).  

\paragraph{Prediction Error}
For our theoretical results on the \GraphModel{} graphs, MPD is given only the expected (as opposed to the actual) degrees. Although this models the uncertainty in the input, it is natural to ask how MPD performs when even the expected degrees are mispredicted. 
To this end, in Appendix~\ref{appendix-noisy}, we suppose that the offline nodes are prioritized in an arbitrary order $\pi'$ which may be different from the order $\pi$ obtained by sorting the nodes according to their expected degrees. Letting $\Delta$ be the minimum number of offline nodes that needs to be deleted such that $\pi$ and $\pi'$ induce the same order on the remaining nodes, we prove that using a noisy degree predictor which induces $\pi'$ instead of $\pi$ can shrink the size of the matching produced by MPD by at most $\Delta$. We note that the number of mispredicted nodes is an upper bound on $\Delta$, but in general $\Delta$ could be much smaller.

%\textcolor{red}{To this end, in Appendix~\ref{appendix-noisy}, we prove that if the offline nodes are prioritized in some order other than sorted in terms of their expected degrees, then the expected size of the matching produced by MPD can shrink by at most the minimum number of offline nodes which would need to be removed for the rest to be in sorted order (this quantity is at most the number of mispredicted nodes).}

Importantly, we also note that the empirical performance of MPD shows its resilience to prediction error. Our experiments on real graphs use predictors which are noisy and which degrade over time but still find large matchings. Furthermore, on synthetic Zipfian data, we experiment with artifically adding noise and show a gradual degradation of MPD's performance as error increases. 

\section{Preliminaries}
\label{sec-prelim}

\paragraph{\GraphModel{}  model}  \GraphModel{} is the bipartite version of the Chung-Lu-Vu model used in prior work~\cite{chung2004spectra,meka2009matrix}. 
Given vectors $\mathbf{p} = \{p_i\}_{i=1}^n$ and $\mathbf{q} = \{q_j\}_{j=1}^m$, the edge $\{u_i, v_j\}$ appears in the graph independently with probability $p_iq_j$. From the vectors $\mathbf{p}$ and $\mathbf{q}$, we obtain the vector of offline expected degrees $\mathbf{d} = \{d_i\}_{i=1}^n=\{p_i\cdot \|\mathbf{q}\|_1\}_{i=1}^n$.
For our theoretical results within this model, our algorithm MinPredictedDegree uses the degree predictor which returns the \emph{expected} degree for each offline node: $\sigma(u_i) = d_i$  (see Appendix~\ref{appendix-noisy} for extension to noisy predictors).
The particular case of symmetric \GraphModel{} where $q=(1,\dots,1)$ corresponds to the case where consumers (online) pick their edges i.i.d.\ over producers (offline) and MPD has knowledge of the average preferences over producers.
\iffalse
\color{red}
Given a vector of offline expected degrees $\mathbf{d} = \{d_i\}_{i=1}^n$, the edge $\{u_i, v_i\}$ appears in the graph independently w.p. $d_i/m$.
    Our algorithm MinPredictedDegree uses the degree predictor which returns the \emph{expected} degree for each offline node: $\sigma(u_i) = d_i$.
    This model corresponds to the case where consumers (online) pick their edges i.i.d.\ over producers (offline) and MinPredictedDegree has knowledge of the average preferences over producers.
\fi

 %\vspace{-0.15 in}   
\paragraph{Known i.i.d. model} In the known i.i.d. model of \cite{feldman2009online}, algorithms are given access to a {\em type graph} $G=(U \cup V,E)$ and a distribution $\mathcal{P}: V \rightarrow [0,1]$.
The nodes in $V$ and their incident edges represent ``types'' of online nodes.
An input instance $\hat{G}=(U \cup \hat{V},\hat{E})$ is formed by picking $m$ online nodes i.i.d. from $V$ according to the probabilities described by $\mathcal{P}$.
Note that the symmetric \GraphModel{} model defined earlier is a special case of this model.
In our experiments, the degree predictions are given by the expected degrees of the offline nodes. %(i.e. the degree of the offline nodes in the type graph rescaled by $|\hat{V}|/|V|$).

\section{Related Work}
\label{sec-related}
Online bipartite matching and its generalizations have been investigated extensively. The survey~\cite{mehta2013online} and the recent paper \cite{borodin2020experimental} provide excellent overviews of this area. The state of the art competitive ratios are $1-1/e \approx 0.632$ in the worst case~\cite{karp1990optimal} and  $\approx 0.7299$ for the known i.i.d. model~\cite{brubach2016new}. See \cite{borodin2020experimental} for an extensive empirical study of the existing algorithms. Other algorithms examined in the experimental section include \cite{feldman2009online, bahmani2010improved, manshadi2012online, jaillet2014online, durr2016power, borodin2019conceptually}.

More generally, there has been lots of interest in online algorithms with predictions over the last few years, for problems like caching~\cite{lykouris2018competitive, rohatgi2020near, wei2020better, jiang2020online}, ski-rental and its generalizations~\cite{purohit2018improving, gollapudi2019online, anand2020customizing, angelopoulos2020online}, scheduling~\cite{mitzenmacher2020scheduling,lattanzi2020online} matching~\cite{kumar2019semi, antoniadis2020secretary, lavastida2020learnable} and learning~\cite{dekel2017online,bhaskara2020online}. Other areas impacted by learning-based algorithms include combinatorial optimization~\cite{khalil2017learning,balcan2018learning,dinitz2021faster},  
similarity search~\cite{salakhutdinov2009semantic, weiss2009spectral, jegou2011product,wang2016learning,dong2020learning}, 
data structures~\cite{kraska2017case,mitz2018model} 
and streaming/sampling algorithms~~\cite{hsu2019learning,jiang2020learning,eden2021learning}.
See \cite{mitzenmacher2020algorithms} for an excellent survey of this area.

\section{Algorithm}
\label{sec-algorithm}
\paragraph{Online Bipartite Matching}
The online bipartite matching problem is defined as follows.
Given a bipartite graph $G = (U \cup V, E)$, we call $U$ the ``offline'' side and $V$ the ``online'' side of the bipartition.
Let $n = |U|$ and $m = |V|$.
The nodes in $U$ are known beforehand and the nodes in $V$ arrive one at a time, along with their incident edges.
An online bipartite matching algorithm maintains a matching throughout the process, with the goal of maximizing the size of the matching.
As each node $v \in V$ arrives, the algorithm can pick one of its neighboring edges to add to the matching.

\paragraph{MinPredictedDegree}
\begin{algorithm}[tb]
  \caption{MinPredictedDegree}
  \label{alg:minpreddegree}
\begin{algorithmic}
  \STATE {\bfseries Input:} Offline nodes $U$ and degree predictor $\sigma: U \rightarrow \mathbb{R}_{\geq 0}$
  \STATE {\bfseries Output:} Matching $M$
  \STATE Initialize $M \leftarrow \emptyset$.
  \WHILE{online node $v \in V$ arrives}
    \STATE $N(v) \leftarrow$ unmatched neighbors of $v$
    \IF{$|N(v)| > 0$}
      \STATE $u^* \leftarrow \argmin_{u \in N(v)} \sigma(u)$ (ties broken arbitrarily)
      \STATE $M \leftarrow M \cup \{(u^*, v)\}$ 
    \ENDIF
  \ENDWHILE
\end{algorithmic}
\end{algorithm}

In addition to knowing the offline nodes $U$ beforehand, MinPredictedDegree (MPD) is given a degree predictor $\sigma: U \rightarrow \mathbb{R}_{\geq 0}$.
In practice, this predictor could be inferred from additional knowledge about the graph or from past data.
When a node $v \in V$ arrives, MPD (see Algorithm~\ref{alg:minpreddegree}) uses this predictor to greedily select the minimum predicted degree neighbor $u^*$ of $v$ that is not already covered in the matching and then adds the edge $\{u^*, v\}$ to the matching.
If no such valid neighbor exists, MPD does nothing with $v$.
Intuitively, low degree offline nodes should be matched as early as possible as they only appear a few times while we will have many chances to match high degree offline nodes.  

The MPD algorithm has similar structure to the worst-case optimal Ranking algorithm~\cite{karp1990optimal} which assigns a random cost to each offline node and at each step greedily matches with the lowest cost offline neighbor.
Specifically, if the degree predictor is random, MPD and Ranking are equivalent. As we show in the later sections, if the predictor is ``good enough'', MPD often performs much better than Ranking, both in theory and in practice. 

%%PIOTR: No need to be specific here, given the intro
\iffalse
In particular, we find experimentally and through an analysis on random power law graphs that MPD competes favorably against baseline algorithms (including Ranking) and often outputs matchings almost as large as the maximum matching. In the worst case (even if the predictions are adversarial), MPD achieves a competitive ratio of $1/2$ as it returns a maximal matching (compared to $1 - 1/e$ achieved by Ranking). We also show that this bound is tight: see Appendix~\ref{appendix-worst-case-bound} for the details.
\fi
%MPD can perform poorly if the degree predictor is arbitrary or if high degree nodes have poor edge expansion compared to low degree nodes (making it disadvantageous to always prioritize low degree nodes).  

%However, often these adversarial structures do not appear in practice and matching low degree nodes first leads to better results.
%In fact, the hard instance given in the seminal paper by Karp, Vazirani, and Vazirani~\cite{karp1990optimal} that introduced the online bipartite matching problem and the Ranking algorithm relies on the fact that algorithms with no extra information on the graph (e.g.\ no degree predictions) must often mistakenly match high degree left nodes rather than low degree left nodes when given a choice. 

%PIOTR: I commented this text out because it felt tangential to the main topic of this section.We might want to move it elsewhere.

\section{Optimality of MPD on \GraphModel{} graphs}\label{sec-opt}
In this section we show that the size of the matching found by the the MPD algorithm stochastically dominates the size of the matching found by any other algorithm. We start by providing some preliminaries for the analysis.
\paragraph{Preliminaries}
For $p\in [0,1]^n$ and $q\in [0,1]^m$, let $I_{p, q}$ denote an instance of a CLV-B graph with $n=|U|$ offline nodes, $m=|V|$ online nodes, and weight vectors $p$ and $q$, such that the probability that an edge $(u_i,v_j)$ exists is equal to $p_iq_j$ for any $i \in [n], j \in [m]$.  Assume with no loss of generality that $p$ is ordered, $p_1 \leq p_2 \leq \ldots \leq p_n$. 
Note that the expected degree of the offline node $u_i$ is $p_i\|q\|_1$, i.e., it is proportional to  the weight $p_i$.

In the online setting, the nodes of $V$ arrive sequentially in the order $v_1,\dots,v_m$ with the random neighborhood of $v_j\in V$ being revealed at the arrival of $v_j$. When $v_j$ arrives, an online bipartite algorithm $A$ can match $v_j$ to any of its unmatched neighbors in $U$ but cannot change its decision later.
For any online bipartite matching algorithm $A$, let $A(I_{p, q})$ denote the size of the matching attained by $A$ on the instance $I_{p,q}$.
Let $A_0$ be the MinPredictedDegree algorithm which matches a node $v_j$ with neighborhood $S$ to an unmatched node $u_i \in S$ such that $p_i$ minimal, i.e.\ to an available node in $S$ with minimal expected degree (ties broken arbitrarily but consistently, e.g.\  by sorted order of the offline node id's). Let $p(A, S)$ be the resulting set of weights after algorithm $A$ (potentially) chooses a neighbor in $S$ to match with.

Consider two ordered weight vectors $p, p'$ both of length $n$.
We say that $p'$ \emph{dominates} $p$, equivalently $p \preceq p'$, if $p_i \leq p'_i$ for all $i \in [n]$.
% \paragraph{Stochastic dominance of MPD}
We are now ready to state our main result on the optimality of MPD.

\begin{theorem}
\label{thm:mpd-opt}
Let $p\in [0,1]^n$ and $q\in[0,1]^m$. Let $A$ be any online algorithm and let $t\geq 0$. Then,
\[
    \Pr(A(I_{p,q}) \geq t) \leq \Pr(A_0(I_{p,q}) \geq t).
\]
\end{theorem}
To prove the theorem, we will need two technical Lemmas. Informally, Lemma~\ref{lem:greedy} states that for any $S\neq \emptyset$, it is an advantage for $A_0$ if the neighborhood of the first arriving node is $S$ rather than the empty set. Lemma~\ref{lem:domination} (the proof of which is the main technical challenge) states that if $p\preceq p'$, then $A_0(I_{p',q})$ stochastically dominates $A_0(I_{p,q})$. Intuitively, Theorem~\ref{thm:mpd-opt} then follows from Lemma~\ref{lem:domination} by inducting on the number of online nodes $m$. For any algorithm $A$ and non-empty subset $S\subseteq[n]$, if $A$ matches $v_1$ to a node in the neighborhood $S$, then $p(A,S)\preceq p(A_0,S)$, and we can apply Lemma~\ref{lem:domination} together with the induction hypothesis with $m-1$ online nodes. We need Lemma~\ref{lem:greedy} to handle the issue that $A$ may not to match $v_1$ even in the case that $S$ is non-empty. The proofs of the two lemmas and of Theorem~\ref{thm:mpd-opt} are postponed to Appendices~\ref{app:greedy}, \ref{app:domination} and \ref{app:mpd-opt}.
\begin{lemma}
\label{lem:greedy}
Let $p\in[0,1]^n$ and $q\in [0,1]^m$ be weight vectors. Let $p^*\in[0,1]^{n-1}$ be obtained from $p$ by removing its $i$'th entry for some $i\in[n]$. For any $t\geq 0$,
\[
\Pr(A_0(I_{p,q}) \geq t) \leq \Pr(A_0(I_{p^*,q}) \geq t-1).
\]
\end{lemma}

%\begin{proof}
%See Appendix~\ref{app:greedy}.
%\end{proof}

\begin{lemma}
\label{lem:domination}
Let $p,p'\in[0,1]^n$, be ordered weight vectors and $q\in[0,1]^m$. Suppose that $p \preceq p'$. For any $t\geq 0$,
\[
    \Pr(A_0(I_{p,q}) \geq t) \leq \Pr(A_0(I_{p', q}) \geq t).
\]
\end{lemma}

While optimally only holds when the predicted degrees are the expected degrees (or at least induce the same ordering over the offline nodes), the performance of MPD cannot be much worse if the predictions are slightly off. Formally, for  an arbitrary degree predictor $\sigma$, let $p[\sigma]$ be the array of \GraphModel{} offline weights ordered by $\sigma$ and let $\LIS(p[\sigma])$ be the size of the longest increasing subsequence in this array. We show (via a more general result) in Appendix~\ref{appendix-noisy} that MPD will match at most $n - \LIS(p[\sigma])$ fewer nodes than when given the expected degrees as predictions.
%\begin{proof}
%See Appendix~\ref{app:domination}.
%\end{proof}

%See Appendix~\ref{app:mpd-opt} for the proof of Theorem~\ref{thm:mpd-opt}.

\section{Competitive ratio of MPD on symmetric \GraphModel{} random graphs}
\label{sec-analysis}
%\Justin{change definition of CLV-B and title}
%Following in a long line of work in the average-case analysis of matching algorithms initiated by~\cite{karp1981maximum}, we analyze MPD under a natural random bipartite graph model we refer to as \GraphModel{}, a bipartite version of the Chung-Yu-Vu random graph model~\cite{chung2004spectra}.
Though we know that MPD is optimal within the \GraphModel{} model, this result does not give explicit competitive ratios for MPD.
In this section we analyze MPD under the symmetric \GraphModel{} model, and derive a set of equations that give a lower bound on MPD's competitive ratio. 
To recap, the symmetric model is parameterized by $n=|U|$, $m=|V|$, and a vector $\mathbf{d} = \{d_i\}_{i=1}^n$ corresponding to the expected degrees of the offline nodes.
Formally, for any $u_i \in U$ and $v_j \in V$, the edge $\{u_i,v_j\}$ appears in the graph with probability $d_i/m$.
%This model corresponds to the case where consumers pick their edges i.i.d.\ with $\mathbf{d}$ describing the distribution over producers.

As in the previous section, we analyze MPD when the degree predictions are given by the expected degrees $\mathbf{d}$.
% Note that in this case, the degree predictor is noisy.
% In particular, for each offline node $u_i \in U$, the \emph{true} degree of $u_i$ is distributed as a Binomial random variable with sample size $m$ and probability $d_i/m$.
Our main results within this model are a set of equations that describe the size of the matching produced by MPD as well as the size of the maximum matching.
% \vspace{-1em}
\begin{itemize}[leftmargin=*]
    \item Given a set of expected degrees $\mathbf{d}$, Equation~\ref{eq-mpd-expectation} models the behavior of MPD on a symmetric \GraphModel{}($\mathbf{d}$) graph.
    We extend these results to the asymptotic case in Appendix~\ref{appendix-asymptotic}, giving the expected matching size as $n, m \rightarrow \infty$ for a given distribution of expected degrees.
    \item Given a set of expected degrees $\mathbf{d}$,in Appendix~\ref{appendix-offline-expectation}, we give an upper bound on the expected size of the maximum matching on a symmetric \GraphModel{}($\mathbf{d}$) graph, and in Appendix~\ref{appendix-asymptotic}, we give the asymptotic equivalent.
    Empirically, we find this upper bound to be close to the maximum matching size when $\mathbf{d}$ follows a power law distribution.
    \item Using these equations, we show that in expectation MPD returns matchings almost as large as the maximum when the expected degrees of the offline nodes follow a power law distribution (see Table~\ref{table-analysis-expcut} and Figure~\ref{fig-analysis-zipf}).
    For both MPD and the maximum matching, we show that the matching sizes are concentrated about their expectations (Appendix~\ref{appendix-mpd-concentration} and~\ref{appendix-offline-concentration}), implying that on these graphs, MPD achieves a large competitive ratio.
\end{itemize}

\subsection{Competitive ratios on power law graphs}

\begin{table}[ht]
\begin{center}
\begin{small}
\begin{sc}
\begin{tabular}{lcccr}
\toprule
Cutoff $\lambda$ & $\alpha=0.5$ & $\alpha=1$ & $\alpha=1.5$ & $\alpha=2$   \\
\midrule
10 & 0.967 & 0.948 & 0.934 & 0.928 \\
100 & 0.998 & 0.986 & 0.958 & 0.937\\
1000 & 1.000 & 0.995 & 0.966 & 0.940 \\
10000 & 1.000 & 0.997 & 0.969 & 0.940 \\
100000 & 1.000 & 0.998 & 0.970 &  0.940 \\
\bottomrule
\end{tabular}
\end{sc}
\end{small}
\end{center}
\caption{Lower bound on the competitive ratio of MPD on symmetric \GraphModel{} graphs with offline expected degrees following a power law with exponential cutoff distribution as $n,m \rightarrow \infty$. The fraction of offline nodes with expected degree $d$ is proportional to $d^{-\alpha} e^{-d/\lambda}$ for $d=\{1,2,...\}$.}
\label{table-analysis-expcut}
% \vspace{-0.2 in}
\end{table}

% In Table~\ref{table-analysis-expcut} (and Figure~\ref{fig-analysis-zipf} in Appendix~\ref{appendix-compratiofig}), we show the competitive ratio of MPD based on our equations for expected size of MPD's matching and for the expected size of the maximum matching.
% Note that in the non-asymptotic case, our equations give an approximation to the true expectation of MPD's matching size due to modeling the process with continuous differential equations.
% For the maximum matching analysis, our equations give an upper bound on the maximum matching size, which translates to a lower bound on our ratio (so we will only underestimate the true performance of MPD).
% For the types of graphs we consider, even for relatively small $n$, the analytic ratios match the empirical ratios we found in Figure~\ref{fig-zipf-comparison}, indicating the approximation/bound error is small for these graphs.

In Table~\ref{table-analysis-expcut}, we show the competitive ratio of MPD on symmetric \GraphModel{} graphs with expected offline degrees following a power law with exponential cutoff distribution~\cite{borodin2020experimental, mitzenmacher2005probability} and with $n, m \rightarrow \infty$.
For $d=\{1,2,...\}$, the fraction of offline nodes with expected degree $d$ is proportional to $d^{-\alpha} e^{-d/\lambda}$ for exponent $\alpha$ and cutoff $\lambda$.
Note that in the asymptotic case, as the sizes of MPD's matching and the maximum matching are concentrated about their expectations (Theorems~\ref{thm-mpd-concentration},~\ref{thm-offline-concentration}), the ratio of expectations is equivalent to the competitive ratio (expectation of ratio).
When the exponent is small or the cutoff is large, MPD achieves a better competitive ratio, with the ratio exceeding $0.99$ when both occur.
When $\alpha=2$, while MPD still achieves a competitive ratio of up to $0.94$, the competitive ratio is not as affected by a larger cutoff as with smaller exponents (the power law factor is already significantly limiting the fraction of offline nodes with large expected degree). The analysis we develop is general and can be used to evaluate MPD on symmetric \GraphModel{} graphs with different parameters than those we have considered.

% In what follows, we develop equations that allowed us to produce these results. The analysis we present is general and can be used to evaluate MPD on symmetric \GraphModel{} graphs with different parameters than those we have considered.

\subsection{Differential equation analysis of MPD}
\label{sec-analysis-minpreddegree}
% In order to analyze MPD, we construct a family of random variables and corresponding differential equations that model the number of unmatched left nodes throughout the running of the algorithm.
Let $Y_d^t$ be the number of offline nodes with expected degree $d$ who are unmatched by MPD after seeing the $t$th online node.
Within this random graph model, $\{Y_d^t\}_{t=0}^m$ form a Markov chain with the following expected evolution:
\begin{equation}
\begin{split}
    \E[Y_d^{t+1} - Y_d^t] =& - \left(1 - \left(1 - d/m\right)^{Y_d^t}\right) \prod_{d' < d} \left(1 - d'/m\right)^{Y_{d'}^t}.
\end{split}
\end{equation}
The first term corresponds to the probability that at least one unmatched offline node with expected degree $d$ is incident on the $(t+1)$st online node while the second term corresponds to the probability that this online node has no neighboring unmatched offline nodes with smaller expected degree (which would be prioritized).

Let $k_d = -\log(1 - d/m)$. To simplify the analysis of MPD, it will be helpful to consider the random variables $Z_d^t = -k_d * Y_d^t$ where
\begin{equation}
\label{eq-transitions2}
\begin{split}
    \E[Z_d^{t+1} - Z_d^t] = k_d \left(1 - e^{Z_d^t}\right) \prod_{d' < d} e^{Z_{d'}^t}.
\end{split}
\end{equation}

Following the work of Kurtz and many subsequent researchers~\cite{kurtz1981approximations, wormald1995differential, mitzenmacher1997studying, luby2001efficient, noiry2021nonasymptotic}, we show that the behavior of MPD as described by these Markov chains is well approximated by the trajectory of the following system of differential equations for all unique expected degrees $d$ in $\mathbf{d}$:
\begin{equation}
\label{eq-diffeq}
    \frac{dz_d(t)}{dt} = k_d \left(1 - e^{z_d(t)}\right) \prod_{d' < d} e^{z_{d'}(t)}.
\end{equation}
These functions $z_d(t)$ represent continuous-time approximations of the Markov chains with their derivatives corresponding to expected change from Equation~\ref{eq-transitions2}.
In Appendix~\ref{appendix-diffeqsoln}, we give the solution to these differential equations.
Relying on past work~\cite{luby2001efficient}, we give the following theorem (see Appendix~\ref{appendix-mpd-expectation} for proof).
\begin{theorem}
\label{thm-mpd-expectation}
Let $G$ be a symmetric \GraphModel{} random graph with unique expected offline degrees $\{\delta_i\}_{i=1}^\ell$.
Let $f_d = \lambda_d \cdot n$ be the number of offline nodes with expected degree $d$.
Then, the expected (over the randomness in $G$) size of the matching formed by MPD approaches 
\begin{equation}
\label{eq-mpd-expectation}
    \sum_{i=1}^\ell f_{\delta_i} + z_{\delta_i}(m)/k
\end{equation}
as $n=m$ approach infinity, where $z_{\delta_i}(t)$ for $i \in \{1,\ldots,\ell\}$ form the solution to the system of differential equations in Equation~\ref{eq-diffeq}.
%In addition, if $Y_d$ is the true number of unmatched offline nodes of expected degree $d$, then there exists a constant $c$ s.t.
%\begin{equation}
% \label{eq-mpd-whp}
%     \Pr(Y_d > -z_{\delta_i}(m)/k + cm^{5/6}) < \ell m^{2/3} \exp(- m^{1/3}/2).
%\end{equation}
\end{theorem}

The solution to the system of differential equations gives us a closed form continuous-time approximation for expected performance of MPD in terms of $\mathbf{d}$.
In particular, in the asymptotic case, the equations give the exact expected performance and in the non-asymptotic case give an approximation on the number of unmatched offline nodes (and thus the matching size).

\section{Experiments}
\label{sec-experiments}
In this section, we evaluate the empirical performance of MPD on real and synthetic data.
For each dataset, we report the empirical competitive ratio of MPD and a variety of baselines.
In each case, the empirical competitive ratio is the average, over 100 trials, of the ratios of the sizes of the matchings outputted by a given algorithm and the sizes of the maximum matching.
In addition to the average ratio, we report one standard deviation of the ratio across the trials.

\paragraph{Datasets}
We evaluate MPD on the following  datasets.
\begin{itemize}[leftmargin=*]
    \item \textbf{Oregon:} $9$ graphs\footnote{
        \label{foot-doublecover}The graphs in the Oregon and CAIDA datasets are made bipartite following the \emph{bipartite double cover} or \emph{duplicating method} used in prior work~\cite{ borodin2020experimental}. Given a graph $G=(V,E)$, the bipartite double cover of $G$ is the graph $G'=(U' \cup V', E')$ where $U'$ and $V'$ are copies of $V$ and there is an edge $\{u'_i, v'_j\} \in E'$ if and only if $\{v_i, v_j\} \in E$.
    }
    sampled over $3$ months representing a communication network of internet routers from the Stanford SNAP Repository~\cite{leskovec2014snap}. Each graph has $\sim 10k$ nodes on each side of the bipartition and $\sim 40k$ edges. For MPD, the offline degree predictor $\sigma: U \rightarrow \mathbb{R}$ is based on the first graph: if an offline node $u$ (i.e. a specific router) appeared in the first graph, $\sigma(u)$ is the degree of $u$ in that graph. If an offline node $u$ did not appear in the first graph, $\sigma(u) = 1$. 
    For each trial, the order of arrival of the online nodes is randomized.

    \item \textbf{CAIDA:} $122$ graphs\footnotemark[\value{footnote}] sampled approximately weekly over $4$ years representing a communication network of internet routers from the Stanford SNAP Repository~\cite{leskovec2014snap}.
    Each graph has $\sim 20k$ nodes on each side of the bipartition and $\sim 100k$ edges.
    The degree predictor is the same as for the Oregon dataset (for each year, the first graph of the year is used to form the predictor).
    As seen in Figure~\ref{fig-oracle-analysis} (see Appendix~\ref{appendix-experiments}), the degree distribution of the graphs for both the Oregon and Caida datasets are long-tailed and the error of the first graph predictor increases over time as the underlying graph evolves. 
    For each trial, the order of arrival of the online nodes is randomized.

    \item \textbf{Symmetric \GraphModel{} random graph:}
    We consider symmetric \GraphModel{} model where the expected offline degrees are distributed according to Zipf's Law, a popular power law distribution where $d_i = C \cdot i^{-\alpha}$~\cite{mitzenmacher2005probability}. 
    %\Piotr{This might be moved earlier, since we also use it for theory.}
    In our experiments, we set size $n=m=1000$, set $C=m/2$, and vary the exponent $\alpha$.
    %The performance of MPD on symmetric \GraphModel{} random graphs is further explored in Section~\ref{sec-analysis}.
    
    \item \textbf{Known i.i.d.:} 
    Finally, we compare MPD to algorithms for the known i.i.d.\ model, copying the methodology of Borodin et al.~\cite{borodin2020experimental} for synthetic power law graphs (Molloy Reed and Preferential Attachment) and real world graphs.
    In the Molloy Reed experiments, the type graph is sample from a family of random graphs with degrees distributed according to a power law with exponential cutoff.
    In the Preferential Attachment experiments, the type graph is formed by the preferential attachment model in which edges are added sequentially with edges between high degree nodes being more likely.
    The Real World graphs are comprised of a variety of graphs from the Network Repository~\cite{rossi2015networkrepository}. See Appendix~\ref{appendix-experiments} for more results on Real World graphs.
\end{itemize}

\paragraph{Baselines} We compare our algorithm to a variety of baseline algorithms.
\begin{itemize}[leftmargin=*]
    \item \textbf{Ranking} In all experiments, we compare to the classic, worst-case optimal Ranking algorithm~\cite{karp1990optimal}.
    
    \item \textbf{MinDegree} The MinDegree algorithm is a version of MPD with a perfect oracle, i.e. $\sigma(u)$ returns the true degree of $u$. In comparison with MPD, MinDegree shows the effect of prediction error on the performance of MPD.
    
    \item \textbf{Known i.i.d.\ baselines} For the experiments in the known i.i.d. case, we also compare to the baselines in the extensive empirical study of~\cite{borodin2020experimental}--see their paper for detailed descriptions of all algorithms. The code is distributed under the GPL license.
    % Most of these baselines rely on the type graph in order to do some preprocessing.
    Notably, the algorithms Category-Advice and 3-Pass are \emph{not} strictly online algorithms: they take multiple passes over the data, using some limited information from previous passes to make better decisions in the next pass.
    It should also be noted that BKPMinDegree is distinct from either the MPD or MinDegree algorithms we have described--it does not use the type graph but rather maintains and updates an estimate of the degree of the offline nodes throughout the runtime of the algorithm.
    
    Most known i.i.d.\ baselines are \emph{not} greedy--they do not always match an online node even if it has unmatched neighbors. \cite{borodin2020experimental} evaluate greedy augmentations of these algorithms (denoted by Algorithm(g)) which match to an arbitrary unmatched neighbor in these cases and generally show them to outperform their non-greedy counterparts. We additionally evaluate MPD augmented versions of these algorithms (denoted by Algorithm(MPD)) which applies the MPD rule in these cases using the expected degrees as predictions.
\end{itemize}

\begin{figure}[ht]
\begin{center}
\centerline{\includegraphics[width=0.45\columnwidth]{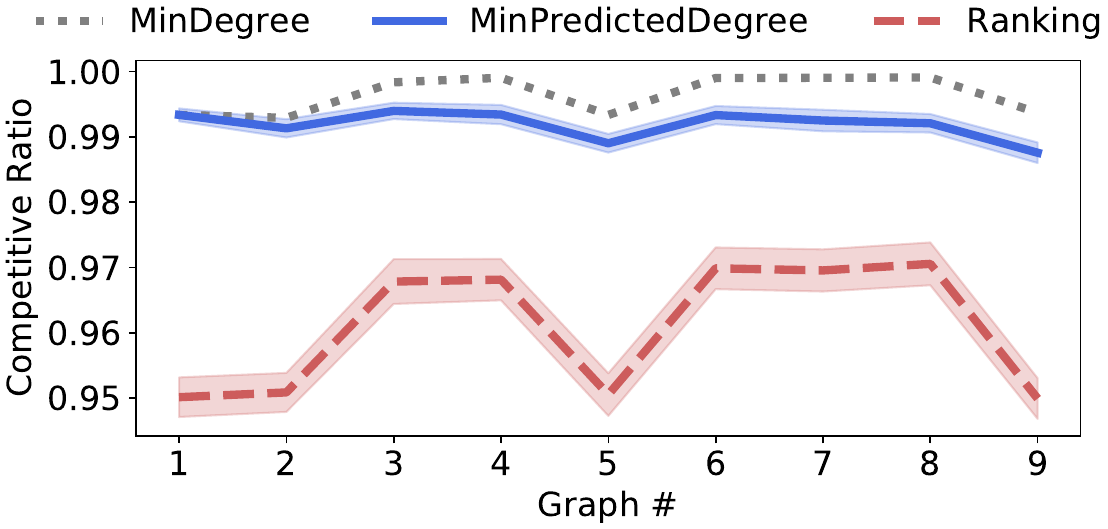}}
\caption{Comparison of empirical competitive ratios on the Oregon dataset. The first graph is used to form  predictions.}
\label{fig-oregon}
% \vspace{-0.25 in}
\end{center}
\end{figure}

\begin{figure}[ht]
\begin{center}
\centerline{\includegraphics[width=\columnwidth]{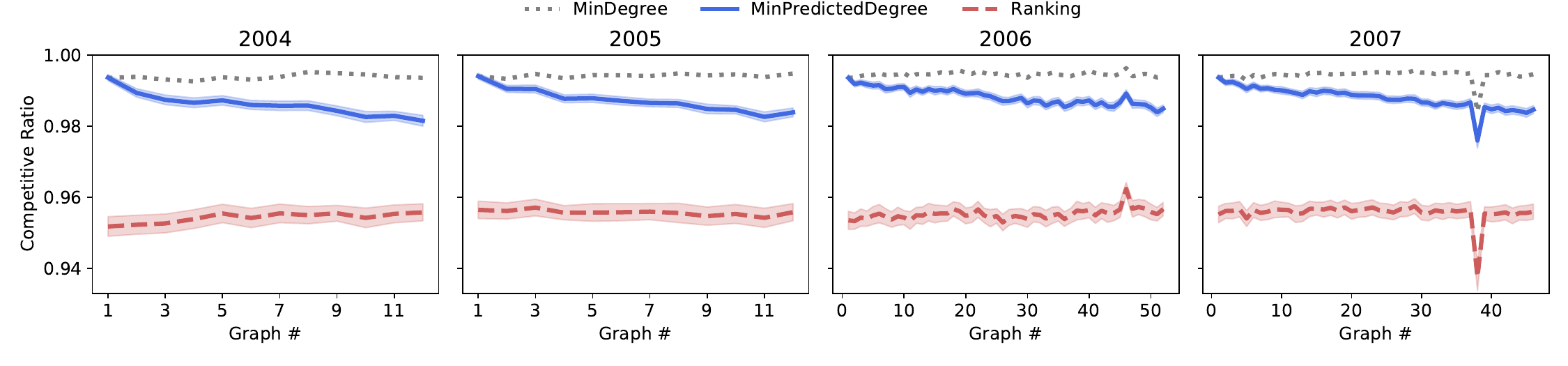}}
\caption{Comparison of empirical competitive ratios on the CAIDA dataset. For each subfigure, the first graph of the year is used to form  predictions for the rest of the year.}
\label{fig-caida}
% \vspace{-0.1 in}
\end{center}
\end{figure}

\begin{figure}[ht]
\centering
\begin{subfigure}[b]{0.4\textwidth}
    \centering
    \includegraphics[width=\textwidth]{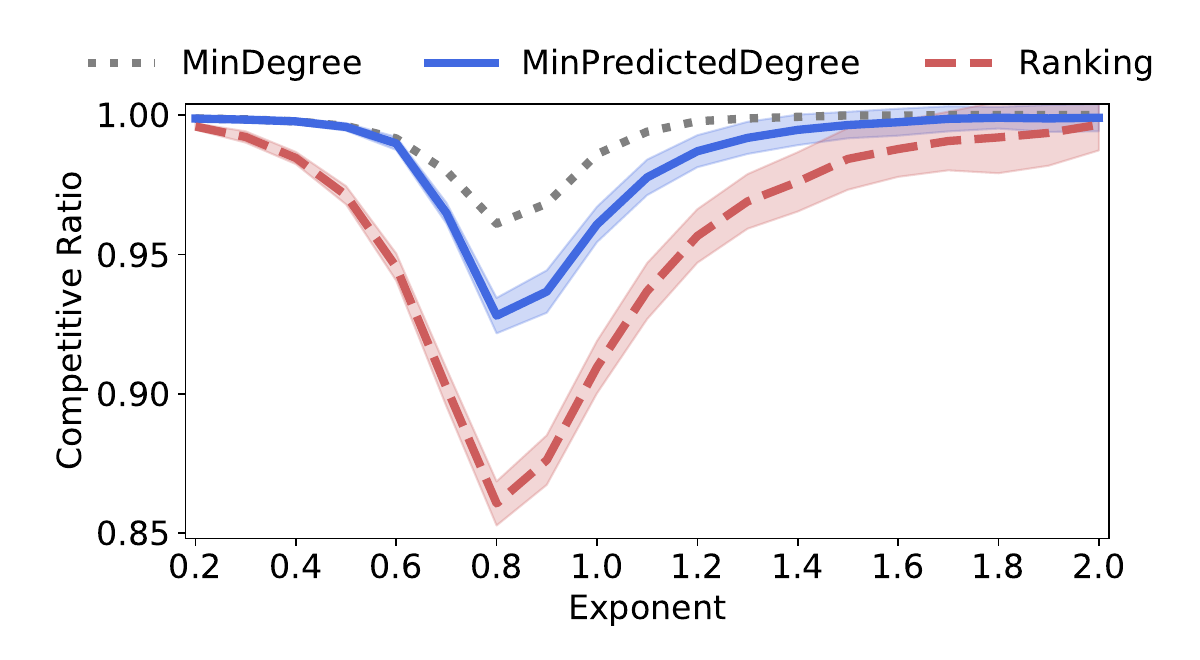}
    \caption{Comparison across Zipf's Law exponents.\\}
    \label{fig-zipf-comparison}
\end{subfigure} \hspace{0.3in}
\begin{subfigure}[b]{0.4\textwidth}
    \centering
    \includegraphics[width=\textwidth]{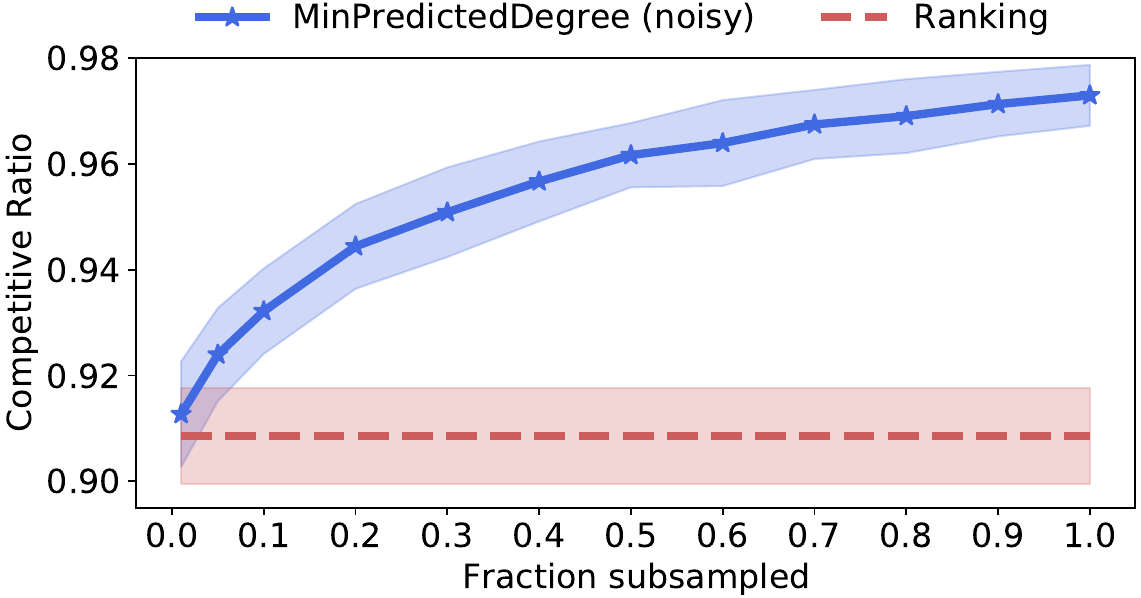}
    \caption{Analysis of predictor noise (exponent $\alpha=1$).}
    \label{fig-zipf-noise}
\end{subfigure}
\caption{Comparison of empirical competitive ratios on symmetric \GraphModel{} random graphs with offline expected degrees following Zipf's Law with exponent $\alpha$. In (a), we vary $\alpha$ and MPD uses the expected degree as its predictor. In (b), the degree predictor is the offline degree in a random subgraph using a (varying) fraction of the online nodes.}
\label{fig-zipf}
% \vspace{-0.1 in}
\end{figure}

\begin{figure}[h!]
\begin{center}
\centerline{\includegraphics[width=\textwidth]{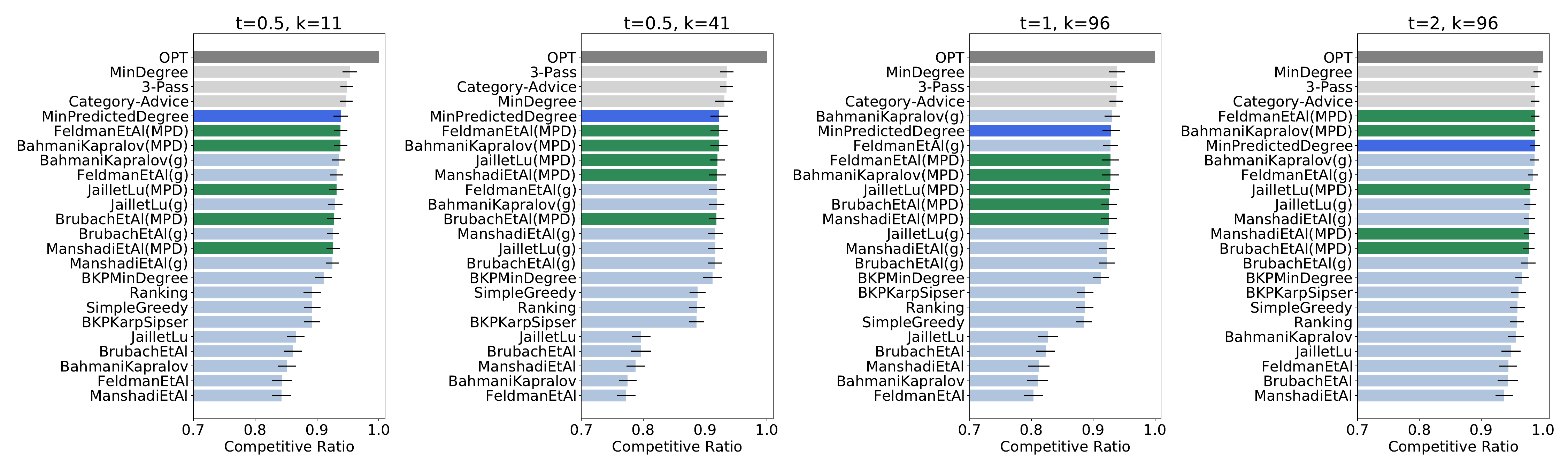}}
\caption{Comparison of empirical competitive ratios for Known i.i.d.\ Molloy-Reed graphs. Algorithms depicted in gray are \emph{not} online algorithms (they use extra information or multiple passes). Algorithms in green are augmented with MPD.}
\label{fig-molloyreed}
% \vspace{-0.25 in}
\end{center}
\end{figure}

\begin{figure}[ht]
\begin{center}
\centerline{\includegraphics[width=\textwidth]{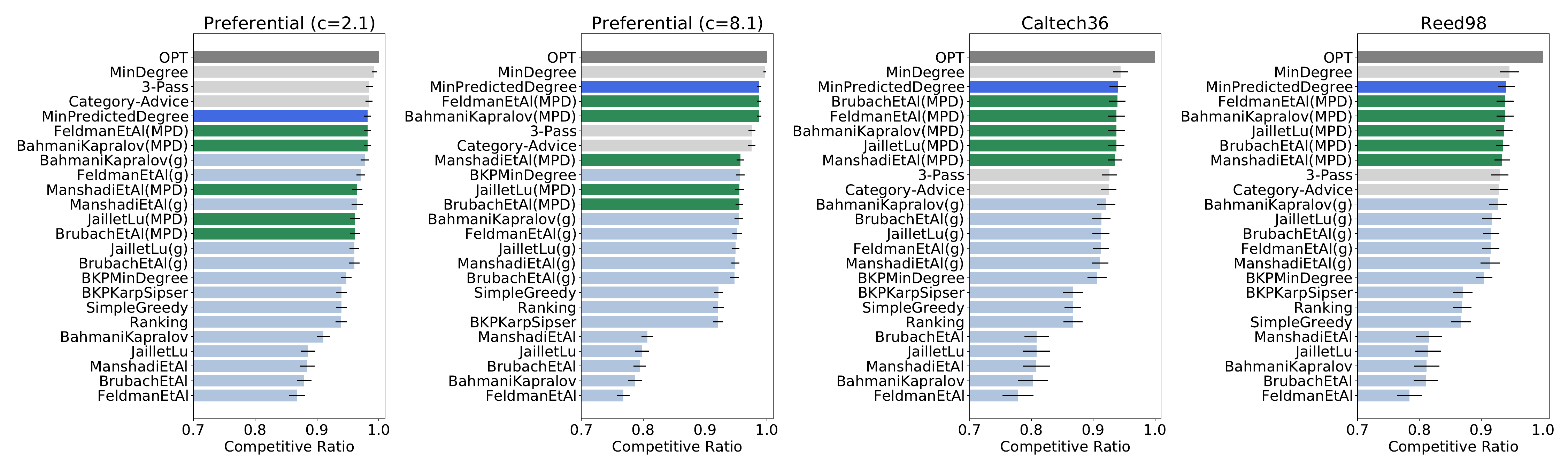}}
\caption{Comparison of empirical competitive ratios on Known i.i.d.\ Preferential Attachment graphs and Real World graphs. Algorithms depicted in gray are \emph{not} online algorithms (they use extra information or multiple passes). Algorithms in green are augmented with MPD. See Appendix~\ref{appendix-experiments} for more Real World results.}
\label{fig-prefANDreal}
% \vspace{-0.25 in}
\end{center}
\end{figure}
 
\paragraph{Results} Across the various datasets, MPD performs well compared to the baseline algorithms.
For the Oregon, CAIDA, and symmetric \GraphModel{} random graph datasets, MPD significantly outperforms Ranking, and for Oregon and CAIDA, the performance of the algorithm mildly declines as the degree predictions degrade.
For the known i.i.d.\ datasets, MPD often outperforms all \emph{online} baselines, despite making only limited use of the known i.i.d.\ model.
Additionally, augmenting the known i.i.d.\ algorithms with the (MPD) rule often improves their performance over both the base and the greedy (g) versions of the algorithms.
\begin{itemize}[leftmargin=*]
    \item \textbf{Oregon and CAIDA (Figures~\ref{fig-oregon},~\ref{fig-caida}):} 
    On the Oregon dataset, MPD achieves a competitive ratio of $\sim 0.99$ across the graphs compared with competitive ratios ranging from $0.95$ to $0.97$ for Ranking.
    Compared with MinDegree, which uses knowledge of the true offline degrees, MPD's performance slowly degrades over time as the graphs become less similar to Graph $\# 1$ (see Figure~\ref{fig-oracle-analysis} in Appendix~\ref{appendix-experiments} for quantitative details).
    % \item \textbf{CAIDA:} 
    
    Similarly, on the CAIDA dataset, MPD does significantly better than Ranking, achieving competitive ratios almost always greater than $0.98$ compared to ratios around $0.95$, respectively.
    As the performance of the degree predictor degrades over time, the performance of MPD gradually declines (though it still significantly outperforms Ranking for both datasets).
    
    \item \textbf{Symmetric \GraphModel{} random graph (Figure~\ref{fig-zipf}):}
    For symmetric \GraphModel{} random graphs with offline expected degrees following Zipf's Law, MPD outperforms Ranking across a spectrum of exponents $\alpha$ ranging from $0.2$ to $2$.
    For exponents less than $0.5$ and greater than $1.5$, MPD achieves a competitive ratio close to $1$ (greater than $0.995$).
    All of the online algorithms have worse competitive ratios when the exponent is closer to one with MPD achieving a ratio of $\sim 0.93$ and Ranking achieving a ratio of $\sim 0.86$ when $\alpha=0.8$.
    Though MPD does worst at $\alpha=0.8$, it also achieves its greatest improvement over Ranking at this setting.
    %Using the equations in Section~\ref{sec-analysis}, in Figure~\ref{fig-analysis-zipf}, we see that for exponents other than $\alpha=1$, the competitive ratio of MPD improves as the size of the graph grows.
    
    In Figure~\ref{fig-zipf-noise}, we analyze the performance of MPD with a noisy degree predictor on Zipf's Law symmetric \GraphModel{} random graphs with exponent 1.
    To introduce noise, the degree predictor $\sigma(u)$ is given by the number of neighbors $u$ has with a random subset of the online nodes $V$.
    As we decrease the fraction of $V$ we subsample, thus increasing the variance of the predictor, the performance of MPD steadily declines.
    Even when the degree predictor only uses 10\% or even 1\% (the leftmost point on the graph) of the online nodes, it still outperforms Ranking.
    
    \item \textbf{Known i.i.d. (Figures~\ref{fig-molloyreed},~\ref{fig-prefANDreal}):} Across all of the experiments in the known i.i.d.\ model, MPD is among the top online algorithms, and is often the best performing online algorithm (note the algorithms in gray are \emph{not} strictly online algorithms).
    Most of the algorithms (e.g.\ BahamiKapralov and ManshadiEtAl) rely heavily on the type graph, including precomputing an optimal matching on the type graph.
    By contrast, MPD only uses first-order information: it only looks at degrees and does not rely on any information about specific edges.
    Even so, in most cases, it outperforms all of the other online algorithms.
    Additionally, the (MPD) augmented versions of the known i.i.d.\ algorithms always beat the base algorithms and often beat the greedy (g) versions, indicating the potential of predicted degrees to be integrated with other algorithms.
    Note that while the standard deviations are quite wide (the known i.i.d.\ model is inherently stochastic), as the results are summarized over 100 trials, relatively small differences in the \emph{average} performance of these algorithms are statistically significant as the standard error is small.
\end{itemize}

\paragraph{Acknowledgements}
This research was supported in part by the NSF TRIPODS program (awards CCF-1740751 and DMS-2022448), NSF award CCF-2006798, Simons Investigator Award, NSF Graduate Research Fellowship under Grant No. 1745302, MathWorks Engineering Fellowship, and DFF-International Postdoc Grant 0164-00022B from the Independent Research Fund Denmark.
\bibliography{arxiv}
\bibliographystyle{plain}

%%%%%%%%%%%%%%%%%%%%%%%%%%%%%%%%%%%%%%%%%%%%%%%%%%%%%%%%%%%%
\section*{Checklist}

%%% BEGIN INSTRUCTIONS %%%
% The checklist follows the references.  Please
% read the checklist guidelines carefully for information on how to answer these
% questions.  For each question, change the default \answerTODO{} to \answerYes{},
% \answerNo{}, or \answerNA{}.  You are strongly encouraged to include a {\bf
% justification to your answer}, either by referencing the appropriate section of
% your paper or providing a brief inline description.  For example:
% \begin{itemize}
%   \item Did you include the license to the code and datasets? \answerYes{See Section~\ref{gen}.}
%   \item Did you include the license to the code and datasets? \answerNo{The code and the data are proprietary.}
%   \item Did you include the license to the code and datasets? \answerNA{}
% \end{itemize}
% Please do not modify the questions and only use the provided macros for your
% answers.  Note that the Checklist section does not count towards the page
% limit.  In your paper, please delete this instructions block and only keep the
% Checklist section heading above along with the questions/answers below.
%%% END INSTRUCTIONS %%%

\begin{enumerate}

\item For all authors...
\begin{enumerate}
  \item Do the main claims made in the abstract and introduction accurately reflect the paper's contributions and scope?
    \answerYes{}
  \item Did you describe the limitations of your work?
    \answerYes{See, for example, Appendix~\ref{appendix-worst-case-bound}.}
  \item Did you discuss any potential negative societal impacts of your work?
    \answerNo{We did not include a discussion of this in the text of the paper as our paper is focused on fundamental research. That being said, as online bipartite matching has a wide variety of applications, developing and better understanding algorithms for this problem may have impact on these diverse applications. In particular, metrics other than the competitive ratio of algorithms (such as certain notions of fairness) are important to consider in applications involving individuals and are being studied in other works.}
  \item Have you read the ethics review guidelines and ensured that your paper conforms to them?
    \answerYes{}
\end{enumerate}

\item If you are including theoretical results...
\begin{enumerate}
  \item Did you state the full set of assumptions of all theoretical results?
    \answerYes{}
        \item Did you include complete proofs of all theoretical results?
    \answerYes{}
\end{enumerate}

\item If you ran experiments...
\begin{enumerate}
  \item Did you include the code, data, and instructions needed to reproduce the main experimental results (either in the supplemental material or as a URL)?
    \answerYes{See supplementary material.}
  \item Did you specify all the training details (e.g., data splits, hyperparameters, how they were chosen)?
    \answerYes{}
        \item Did you report error bars (e.g., with respect to the random seed after running experiments multiple times)?
    \answerYes{}
        \item Did you include the total amount of compute and the type of resources used (e.g., type of GPUs, internal cluster, or cloud provider)?
    \answerNo{The experiments are relatively lightweight and the focus of this paper is not on computational resources (rather, we focus on the quality of the algorithms' output). That being said, MPD is a very simple and computationally efficient algorithm (assuming accessing the degree predictor is not too expensive). The experiments were all run on a 2018 MacBook Pro.}
\end{enumerate}

\item If you are using existing assets (e.g., code, data, models) or curating/releasing new assets...
\begin{enumerate}
  \item If your work uses existing assets, did you cite the creators?
    \answerYes{}
  \item Did you mention the license of the assets?
    \answerYes{}
  \item Did you include any new assets either in the supplemental material or as a URL?
    \answerYes{See code in supplement.}
  \item Did you discuss whether and how consent was obtained from people whose data you're using/curating?
    \answerNA{All data was gathered from open sources intended for research use.}
  \item Did you discuss whether the data you are using/curating contains personally identifiable information or offensive content?
    \answerNA{None of our data falls within this category.}
\end{enumerate}

\item If you used crowdsourcing or conducted research with human subjects...
\begin{enumerate}
  \item Did you include the full text of instructions given to participants and screenshots, if applicable?
    \answerNA{}
  \item Did you describe any potential participant risks, with links to Institutional Review Board (IRB) approvals, if applicable?
    \answerNA{}
  \item Did you include the estimated hourly wage paid to participants and the total amount spent on participant compensation?
    \answerNA{}
\end{enumerate}

\end{enumerate}

%%%%%%%%%%%%%%%%%%%%%%%%%%%%%%%%%%%%%%%%%%%%%%%%%%%%%%%%%%%%

\appendix

\newpage
\section{The Proof of Lemma~\ref{lem:greedy}}\label{app:greedy}

In this appendix, we provide the proof of Lemma~\ref{lem:greedy}. 
\begin{proof}[Proof of Lemma~\ref{lem:greedy}]
The result follows by coupling the instances $I_{p,q}$ and $I_{p^*,q}$. Let $U=(u_1,\dots,u_n)$, $V=(v_1,\dots,v_m)$, and $U^*=(u_1,\dots u_{i-1},u_{i+1},\dots,u_n)$. 
For any instance of a bipartite graph $I$ on $(U,V)$ we can generate an instance $I^*$ on $(U^*,V)$ as follows: For each $j\in[m]$, if $N(v_j)$ is the neigborhood of $v_j$ in $I$, then the neighborhood of $v_j$ in $I^*$ is $N(v_j)\cap U^*$. It is readily checked that if $I$ is distributed as $I_{p,q}$, then $I^*$ is distributed as $I_{p^*,q}$. To prove the result, it therefore suffices to show that for any instance $I$, if $A_0(I)\geq t$, then $A_0(I^*)\geq t-1$, or equivalently, that $ A_0(I^*)\geq A_0(I)-1$. We prove this by induction on $m$. The case $m=0$ is trivial, so assume that $m>0$ and inductively that the result holds for smaller values of $m$. Let $M_j\subseteq U$ denote the set of matched nodes by MPD on $I$ after the arrival of $v_1,\dots,v_j$ and let similarly $M_j^*\subseteq U^*$ denote the set of matched nodes by MPD on $I^*$ after the arrival of $v_1,\dots,v_j$.

We will proceed by cases. First, consider the case in which MPD does not match any vertex $v_j$ to $u_i$ during its run on $I$. Then MPD's behavior on $I$ and $I^*$ will be identical as other than $u_i$ and its incident edges, the two graphs are the same. In particular, $A_0(I)=|M_m|=|M_m^*|=A_0(I^*)$.

Consider next the case in which MPD does match $u_i$ to some vertex $v_j$ during its run on $I$. Up until the arrival of $v_j$, MPD has behaved similarly on $I$ and $I^*$, and in particular, $M_{j-1}=M_{j-1}^*$. Now if the neighborhood of $v_j$ in $U^*$ does not contain an unmatched node, then $M_{j}=M_{j}^*\cup\{u_i\}$, and from this point on, MPD will behave similarly on the two instances $I$ and $I^*$. In particular, $A_0(I)=|M_m(I)|=|M_m^*(I)\cup \{u_i\}|=A_0(I^*)+1$, as desired. If on the other hand, MPD on $I^*$ matches $v_j$ to some node $u_{i'}\neq u_{i}$ in $U^*$, then $|M_j|=|M_j^*|$. Moreover, by the induction hypothesis, $|M_m\setminus M_j|\leq |M_m^*\setminus M_j^*|+1$. It follows that $A_0(I)=|M_j|+|M_m\setminus M_j|\leq |M_j^*| + |M_m^*\setminus M_j^*|+1=A_0(I^*)+1$, as desired.
\end{proof}

\section{The Proof of Lemma~\ref{lem:domination}}\label{app:domination}

In this appendix, we provide the proof of Lemma~\ref{lem:domination}. We start with some preliminaries.
For some subset of offline nodes $S \subseteq [n]$, let $\Pr_p(N_S)$ be the probability that the first online node's neighborhood is exactly the set $S$ in $I_{p,q}$. Note that $\Pr_p(N_S)$ depends only on $p$ and $q_1$ and equals $\prod_{i\in S} (p_iq_1)\prod_{i\in [n]\setminus S}(1-p_iq_1)$.
Let $p(A, S)$ be the ordered vector of weights of the unmatched offline nodes remaining after running $A$ on the first online node if this node has neighborhood $S$. If $A=A_0$, then $p(A,S)$ is obtained by removing the entry $p_i$ of $p$ corresponding to the $u_i\in S$ with minimal $p_i$ (if $S=\emptyset$, then $p(A,S)=p$), but generally, $A$ could behave differently even choosing not to match $v_1$ to any node in $S\neq \emptyset$.
\begin{proof}[Proof of Lemma~\ref{lem:domination}]
We will prove the lemma by induction on the number of online nodes, $m$.
The base case $m = 0$ is trivial. Indeed, in this case, $A_0(I_{p,q}) =A_0(I_{p',q}) = 0$ with probability $1$, so the probabilities of attaining matching size at least $t$ are equal for the two weight vectors $p$ and $p'$.

Now, we will consider the inductive case. Consider any $m > 0$ and assume the statement holds for $m-1$ online nodes. Define the $(m-1)$-dimensional vector $q^*=(q_2,\dots,q_m)$.
Then,
\begin{align}\label{eq:ngbh_expansion}
    \Pr(A_0(I_{p, q}) \geq t) =  \Pr_p(N_\emptyset) \Pr(A_0(I_{p, q^*}) \geq t) + \sum_{S \subseteq [n], S \neq \emptyset} \Pr_p(N_S) \Pr(A_0(I_{p(A_0, S),q^*}) \geq t-1),
\end{align}
and a similar identity holds with $p$ replaced by $p'$.

Denote by $r=q_1p=(q_1p_1,\dots,q_1p_n)$, and $r'=q_1p'=(q_1p_1',\dots,q_1p_n')$, so that $r_i$ and $r_i'$ are the probabilities that $v_1$ has an edge to $u_i$ and $u_i'$ in respectively $I_{p,q}$ and $I_{p',q}$. As $p\preceq p'$, also $r\preceq r'$. In particular, we can write $1-r_i'=(1-s_i)(1-r_i)$ for some $s_i\in [0,1]$. For $i\in [n]$, we let $X_i$ and $Y_i$ be independent Bernoulli variables with $\Pr[X_i=1]=r_i$ and $\Pr[Y_i=1]=s_i$. Let furthermore $Z_i$ be the Bernoulli variable which is $1$ if either $X_i=1$ or $Y_i=1$, and zero otherwise. Then $\Pr[Z_i=1]=r_i'$. We now let $N=\{i \in [n] \mid X_i=1\}$ and $N'=\{i \in [n] \mid  Z_i=1\}$, noting that $N\subseteq N'$. 
For any $T\subseteq [n]$ we can then write
\begin{align*}
    \Pr_{p'}(N_T)=\Pr[N'=T]= \sum_{S \subseteq T}\Pr[N=S] \Pr[N'=T\mid N=S]= \sum_{S \subseteq T}\Pr_p(N_S) \Delta(S,T),
\end{align*}
where we have put $\Pr[N'=T\mid N=S]=\Delta(S,T)$. We note for later use that for any $S\subseteq [n]$, it holds that $\sum_{T\supseteq S} \Delta(S,T)=1$. Indeed, conditioned on $N=S$, it holds that $S\subseteq N'$ with probability 1.
Combining this with~\eqref{eq:ngbh_expansion}, we can write the probability of MPD exceeding size $t$ on the graphs parameterized by $p'$ as
\begin{align}\label{eq:big_bound}
    \Pr(A_0(I_{p',q}) \geq t) =&  \Pr_{p'}(N_\emptyset) \Pr(A_0(I_{p', q^*}) \geq t) + \sum_{T \subseteq [n], T \neq \emptyset} \Pr_{p'}(N_T) \Pr(A_0(I_{p'(A_0, T), q^*}) \geq t-1) \nonumber \\ 
    =& \Pr_{p'}(N_\emptyset) \Pr(A_0(I_{p',q^*}) \geq t) + \nonumber\\
    & \sum_{T \subseteq [n], T \neq \emptyset} \sum_{S \subseteq T} \Pr_{p}(N_S)\Delta(S,T) \Pr(A_0(I_{p'(A_0, T), q^*}) \geq t-1) \nonumber \\
    =& \Pr_{p}(N_\emptyset) \Delta(\emptyset,\emptyset) \Pr(A_0(I_{p', q^*}) \geq t) + \nonumber\\
    & \Pr_p(N_\emptyset) \sum_{T \subseteq [n], T \neq \emptyset} \Delta(\emptyset,T) \Pr(A_0(I_{p'(A_0, T), q^*}) \geq t-1) + \nonumber \\
    & \sum_{S \subseteq [n], S \neq \emptyset}  \Pr_{p}(N_S) \sum_{T \supseteq S} \Delta(S,T) \Pr(A_0(I_{p'(A_0, T), q^*}) \geq t-1),
\end{align}
where the final steps follows by interchanging summations in the second term, and splitting into the cases $S=\emptyset$ and $S\neq \emptyset$.

Note that if $S$ is non-empty and $S\subseteq T$, then by the MPD rule, $p'(A_0, S) \preceq p'(A_0, T)$ as the minimum degree within the set of neighbors cannot be larger in $T$ than in $S$. Moreover, it is readily checked that the assumption $p\preceq p'$ implies that $p(A_0, S) \preceq p'(A_0, S)$ (after an appropriate permutation). Using the induction hypothesis, we get that for $S\neq \emptyset$ and $T\supseteq S$,
\begin{align}\label{eq:inductive1}
    \Pr(A_0(I_{p'(A_0, T), q^*}) \geq t-1) \geq \Pr(A_0(I_{p(A_0, S), q^*}) \geq t-1),
\end{align}
and that 
\begin{align}\label{eq:inductive2}
\Pr(A_0(I_{p',q^*})\geq t)\geq \Pr(A_0(I_{p,q^*})\geq t).
\end{align}
Moreover, an application of Lemma~\ref{lem:greedy} and the induction hypothesis gives that
\begin{align}\label{eq:the_lemma}
    \Pr(A_0(I_{p'(A_0, T), q^*}) \geq t-1)\geq \Pr(A_0(I_{p', q^*}) \geq t)\geq \Pr(A_0(I_{p, q^*}) \geq t).   
\end{align}
Plugging the bounds of~\eqref{eq:inductive1},\eqref{eq:inductive2}, and \eqref{eq:the_lemma} into~\eqref{eq:big_bound}, it follows that
\begin{align*}
    \Pr(A_0(I_{p', q}) \geq t) \geq& \Pr_{p}(N_\emptyset) \Delta(\emptyset,\emptyset) \Pr(A_0(I_{p, q^*}) \geq t) +\\
    & \Pr_p(N_\emptyset) \sum_{T \subseteq [n], T \neq \emptyset} \Delta(\emptyset,T) \Pr(A_0(I_{p, q^*}) \geq t) + \\
    & \sum_{S \subseteq [n], S \neq \emptyset}  \Pr_{p}(N_S) \sum_{T \supseteq S}\Delta(S,T) \Pr(A_0(I_{p(A_0, S), q^*}) \geq t-1).
\end{align*}
Now combining the first two terms above and using that  $\sum _{T\supseteq S}\Delta(S,T)=1$ for any $S\subseteq [n]$, we obtain that
\begin{align*}
    \Pr(A_0(I_{p', q}) \geq t)\geq& \Pr_{p}(N_\emptyset)\Pr(A_0(I_{p, q^*}) \geq t)+\sum_{S \subseteq [n], S \neq \emptyset}  \Pr_{p}(N_S)\Pr(A_0(I_{p(A_0, S), q^*}) \geq t-1)\\
    =&\Pr(A_0(I_{p,q})\geq t),
\end{align*}
where the final equality follows from~\eqref{eq:ngbh_expansion}. This is the desired result.
\end{proof}

\section{The Proof of Theorem~\ref{thm:mpd-opt}}\label{app:mpd-opt}
In this appendix, we provide the proof of Theorem~\ref{thm:mpd-opt}. As in Appendix~\ref{app:domination}, for subsets of offline nodes $S \subseteq [n]$, we let $\Pr_p(N_S)$ denote the probability that the first online node's neighborhood is exactly the set $S$ in $I_{p,q}$.
\begin{proof}[Proof of Theorem~\ref{thm:mpd-opt}]
We will prove the theorem by induction over $m$. 
For the base case, $m = 0$, the inequality holds with equality as both algorithms yield empty matchings.

For the inductive case, let $m > 0$ and assume the inequality holds with $m - 1$ online nodes. Let $q^*=(q_2,\dots,q_m)$.
Let $R_A \subseteq \mathcal{P}([n])$ be set of subsets of $[n]$ s.t.\ for any $S \in R_A$, $p(A, S) = p$, i.e.\ algorithm $A$ matches no edges if the set of neighbors of the current online node is $S$.
Consider the probability of algorithm $A$ attaining a matching with size at least $t$. By the induction hypothesis,
\begin{align*}
    \Pr(A(I_{p, q})\geq t) &= \sum_{S \in R_A} \Pr_p(N_S) \Pr(A(I_{p,q^*}) \geq t) + \sum_{S \subseteq [n], S \notin R_A} \Pr_p(N_S) \Pr(A(I_{p(A, S), q^*}) \geq t-1) \\
    &\leq \sum_{S \in R_A} \Pr_p(N_S) \Pr(A_0(I_{p, q^*}) \geq t) + \sum_{S \subseteq [n], S \notin R_A} \Pr_p(N_S) \Pr(A_0(I_{p(A, S), q^*}) \geq t-1).
\end{align*}
Now if $S\in R_A$, then $p(A,S)\preceq p(A_0,S)$. Thus, by applying Lemma~\ref{lem:greedy} to the first term and Lemma \ref{lem:domination} to the second term above,
\begin{align*}
    \Pr(A(I_{p, q})\geq t) \leq& \Pr_p(N_\emptyset) \Pr(A_0(I_{p, q^*}) \geq t) + \sum_{S \in R_A \setminus \{\emptyset\}} \Pr_p(N_S) \Pr(A_0(I_{p(A_0, S), q^*}) \geq t-1) +\\
    &\sum_{S \subseteq [n], S \notin R_A} \Pr_p(N_S) \Pr(A_0(I_{p(A_0, S), q^*}) \geq t-1) \\
    =& \Pr_p(N_\emptyset) \Pr(A_0(I_{p, q^*}) \geq t) + \sum_{S \subseteq[n], S \neq \{\emptyset\}} \Pr_p(N_S) \Pr(A_0(I_{p(A_0, S), q^*}) \geq t-1) \\
    =& \Pr(A_0(I_{p, q^*}\geq t),
\end{align*}
completing the proof.
\end{proof}

\section{MPD with Noisy Predictions}
\label{appendix-noisy}
In Section~\ref{sec-opt}, we show that given expected degrees as predictions, MPD is the optimal algorithm on \GraphModel{} graphs. Here, we extend that analysis to show that even if the predictor is noisy, MPD can still return a large matching.

Consider two degree predictors $\sigma$ and $\sigma'$ (or more generally, two orderings of the offline nodes). Let $\Delta(\sigma, \sigma')$ denote the minimum number of offline nodes which must be removed such that $\sigma$ and $\sigma'$ induce the same ordering over the remaining nodes.

\begin{theorem}
For any graph $G = (U \cup V, E)$, the matching returned by MPD with degree predictor $\sigma$ has at most $\Delta(\sigma, \sigma')$ more edges than the matching returned by MPD with degree predictor $\sigma'$.
\label{thm-LISerror}
\end{theorem}

On \GraphModel{} graphs, let $p[\sigma]$ be the array of offline weights ordered by $\sigma$. So, if $\sigma^*$ returns the expected degrees, $p[\sigma^*]$ is in sorted order. Let $\LIS(p)$ denote the size of the longest increasing subsequence of $p$. Note that $\Delta(\sigma^*, \sigma) = n - \LIS(p)$. Then, as a corollary to Theorem~\ref{thm-LISerror}, we can bound the performance of a noisy degree predictor compared to the performance of MPD with expected degrees.
\begin{corollary}
On \GraphModel{} graphs, the expected size of the matching returned by MPD with degree predictor $\sigma$ has at most $n - LIS(p[\sigma])$ fewer nodes than the expected size of the matching returned by MPD given the expected degrees.
\end{corollary}
As MPD with expected degrees is the optimal online algorithm for \GraphModel{} graphs, this implies that as long as $n-LIS(p[\sigma])$ is small, MPD is still near-optimal.
Note that $n - LIS(p[\sigma])$ is equal to the minimum number of offline nodes that would need to be removed s.t.\ the offline weights are in sorted order. This quantity is clearly upper bounded by the number of mispredicted nodes as removing each mispredicted node will leave a remaining sequence which is sorted.

In order to prove the theorem, we will use the following lemmas.
% Lemma~\ref{lem:greedy} as well as the following lemma that removing an offline node will never increase the size of the matching produced by any permutation-based algorithm (MPD with an arbitrary degree predictor).
\begin{lemma}[Corollary of Lemma 2 of~\cite{birnbaum2008simple}]
\label{lem:removenode1}
Consider any graph $G = (U \cup V, E)$, offline node $u \in U$, and degree predictor $\sigma$. MPD with degree predictor $\sigma$ when run on $G' = (U \setminus \{u\} \cup V, E)$ will produce a matching with at most one fewer edge than when run on $G$.
\end{lemma}
This follows by~\cite{birnbaum2008simple} as these two matchings will differ by at most one alternating path, which can reduce the matching size by at most one.

\begin{lemma}
\label{lem:removenode2}
Consider any graph $G = (U \cup V, E)$, offline node $u \in U$, and degree predictor $\sigma$. MPD with degree predictor $\sigma$ when run on $G' = (U \setminus \{u\} \cup V, E)$ will produce a matching no bigger than when run on $G$.
\end{lemma}

\begin{proof}
Let $N_G(v)$ be the neighborhood of a node $v \in V$ in the graph $G$, only including \emph{unmatched} offline nodes at the point that $v$ arrives. We will prove via induction that for any online node $v$, $N_{G'}(v) \subseteq N_G(v)$.
For the base case, consider the first online node $v_1$.
In this case, if $u \notin N_G(v_1)$, $N_{G'}(v_1) = N_G(v_1)$. Otherwise $N_{G'}(v_1) = N_G(v_1) \setminus \{u\}$.

Now, consider the inductive case for some online node $v_i$ for $i > 1$. Assume that for all $j < i$, $N_{G'}(v_j) \subseteq N_{G}(v_j)$. Assume, for the sake of contradiction, that there exists some $u'$ s.t.\ $u' \in N_{G'}(v_i)$ but $u' \notin N_G(v_i)$.
This can only happen if at some point previously, $u'$ was matched by MPD running on $G$ but was never matched by MPD running on $G'$. Let $v'$ be the online node that matched with $u'$ when run on $G$. Note that as $u' \in N_{G'}(v_i)$, it must be the case that $u' \in N_{G'}(v')$ as MPD running on $G'$ has not yet matched $u'$ and $u'$ was in $N_G(v')$.
Let $u^*$ be the node that $v'$ matched with in $G'$.
As MPD had the choice to match $v'$ with $u'$ or $u^*$ on $G'$, it must be the case that $\sigma(u^*) < \sigma(u')$ (breaking ties arbitrarily but consistently).
On the other hand, by the inductive hypothesis, $u', u^* \in N_{G}(v')$, so MPD run on $G$ chose to match $v'$ with $u'$ over $u^*$, so $\sigma(u') < \sigma(u^*)$, and we have reached a contradiction.

To complete the proof, note that for MPD (or any greedy algorithm), the size of the matching is exactly $m$ minus the number of times an online node has no available neighbors. The inductive statement implies that for all online nodes $v$, $|N_{G'}(v)| \leq |N_G(v)|$. Therefore, the number of times an online node will have no available neighbors is at least as large for MPD when run on $G'$ as on $G$.
\end{proof}

\begin{proof}[Proof of Theorem~\ref{thm-LISerror}]
Let $S \subset U$ be the $\Delta(\sigma, \sigma')$ online nodes which, if removed, would leave the remaining offline weights in sorted order. Let $A_\sigma(G)$ refer to size of matching returned by MPD with degree predictor $\sigma$ on graph $G$. Let $G_{-S} = \left((U \setminus S) \cup V, E\right)$.

By repeated invocation of Lemma~\ref{lem:removenode2}, 
\[
    A_{\sigma'}(G) \geq A_{\sigma'}(G_{-S})
\]
as removing an offline node never increases the size of the matching returned by MPD.
By repeated invocation of Lemma~\ref{lem:removenode1}, 
\[
    A_\sigma(G) \leq A_\sigma(G_{-S}) + \Delta(\sigma, \sigma')
\]
as removing an offline node can decrease the size of the matching returned by MPD by at most $1$.

To combine these bounds, note that $A_\sigma(G_{-S}) = A_{\sigma'}(G_{-S})$ as $\sigma$ and $\sigma'$ induce the same ordering after removing the offline nodes in $S$.
Therefore,
\[
    A_{\sigma'}(G) \geq A_\sigma(G) - \Delta(\sigma, \sigma'),
\]
completing the proof.
\end{proof}

\section{Worst-Case Bound and Failure Modes}
\label{appendix-worst-case-bound}
In the worst-case, MinPredictedDegree achieves a competitive ratio of $1/2$.
When $\sigma$ gives arbitrary predictions, MinPredictedDegree is equivalent to the simple greedy algorithm for online bipartite matching which achieves a competitive ratio of $1/2$~\cite{karp1990optimal, mehta2013online}.
Even if $\sigma$ is the perfect degree predictor where $\sigma(u) = deg(u)$ for all $u \in U$, MinPredictedDegree is still $(1/2)$-competitive in the worst-case.

As MinPredictedDegree forms a \emph{maximal} matching, the matching it returns is always at least half the size of the maximum matching.
For the matching upper bound on the competitive ratio in the perfect predictor case, consider the graph $G = (U \cup V, E)$ where $n=m=6$ with the following adjacency list:
\begin{align*}
    \{u_1:&\; v_1, v_2, v_3\}, \\
    \{u_2:&\; v_1, v_2, v_3\}, \\
    \{u_3:&\; v_1, v_2, v_3\}, \\
    \{u_4:&\; v_1, v_4\}, \\
    \{u_5:&\; v_2, v_5\}, \\
    \{u_6:&\; v_3, v_6\}.
\end{align*}
The first half of the offline and online nodes form a complete bipartite subgraph while the second half of the offline nodes each connect to one online node in the first half and one online node in the second half.
MinPredictedDegree with a perfect degree predictor will return the matching $M = \{\{u_4, v_1\}, \{u_5, v_2\}, \{u_6, v_3\}\}$, matching the first half of the online nodes with their lower degree neighbors in the second half of the offline nodes, therefore leaving the first half of the offline nodes unmatched.

More generally, MinPredictedDegree can perform poorly if the degree predictor is arbitrary or if high degree nodes have poor edge expansion compared to low degree nodes (making it disadvantageous to always prioritize low degree nodes as seen in the example above).  
However, often these adversarial structures do not appear in practice and matching low degree nodes first leads to better results.
In fact, the hard instance given in the seminal paper by Karp, Vazirani, and Vazirani~\cite{karp1990optimal} that introduced the online bipartite matching problem and the Ranking algorithm relies on the fact that algorithms with no extra information on the graph (e.g.\ no degree predictions) must often mistakenly match high degree left nodes rather than low degree left nodes when given a choice. 

\paragraph{Better worst-case algorithms with degree predictions?} As demonstrated in~\cite{karp1990optimal}, no algorithm for the online matching problem has a competitive ratio better than $1-1/e$. It is however a natural question whether an algorithm with access to the offline degrees can obtain a better worst-case competitive ratio. We answer this question in the negative by modifying the example in~\cite{karp1990optimal} to show that no algorithm with knowledge of the true offline degrees has a competitive ratio of more than $1-1/e$. The hard example in~\cite{karp1990optimal} is an $n$ by $n$ bipartite graph with offline nodes $U=\{u_1,\dots,u_n\}$ and online nodes $V=\{v_1,\dots, v_n\}$. Node $v_i$ has an edge to each node in $\{u_j:i\leq j\leq n\}$. The maximum matching clearly has size $n$ but it is shown in~\cite{karp1990optimal} that if the online nodes arrive in a random order, then no algorithm matches more than $(1-1/e+o(1))n$ nodes in expectation. Note that given the true offline degrees as predictions, MPD will actually achieve the maximum matching for this example. 

We augment this result in a black box manner to construct a bad example even when the offline degrees are known to the algorithm. Assume with no loss of generality that $n$ is a square. Let $U_1,\dots, U_{\sqrt{n}}$ and $V_1,\dots,V_{\sqrt{n}}$ be disjoint vertex sets each of size $\sqrt{n}$. Let the offline nodes be $U=\bigcup_{i\leq\sqrt{n}}U_i$ and the online nodes be $V=\bigcup_{i\leq\sqrt{n}}V_i$. For each $1\leq i\leq \sqrt{n}-1$, we let $(U_i,V_i)$ form an instance of the hard graph from~\cite{karp1990optimal} (as described above) with $\sqrt{n}$ online/offline nodes. Moreover, $(U_{\sqrt{n}},V_{\sqrt{n}})$ form a complete bipartite graph. Finally, for each offline node $u\in U$, we add $\sqrt{n} - deg(u)$ edges from $u$ to $V_{\sqrt{n}}$. Then all of the offline degrees are the same (they are all $\sqrt{n}$), so the degree oracle is of no use. If we sequentially for $1\leq i\leq \sqrt{n}-1$ let the nodes of $V_i$ arrive in a random order (and finally the nodes from $V_{\sqrt{n}}$ in any order), it follows from the result in~\cite{karp1990optimal} that the expected size of the produced matching is at most
$$
(\sqrt{n}-1)(1-1/e+o(1))\sqrt{n}+\sqrt{n}=n(1-1/e+o(1))
$$
As the maximum matching has size $n$, the upper bound on the competitive ratio follows. 

\section{Additional Competitive Ratio Results}
\label{appendix-compratiofig}

\begin{figure}[h!]
\begin{center}
\centerline{\includegraphics[width=0.4\columnwidth]{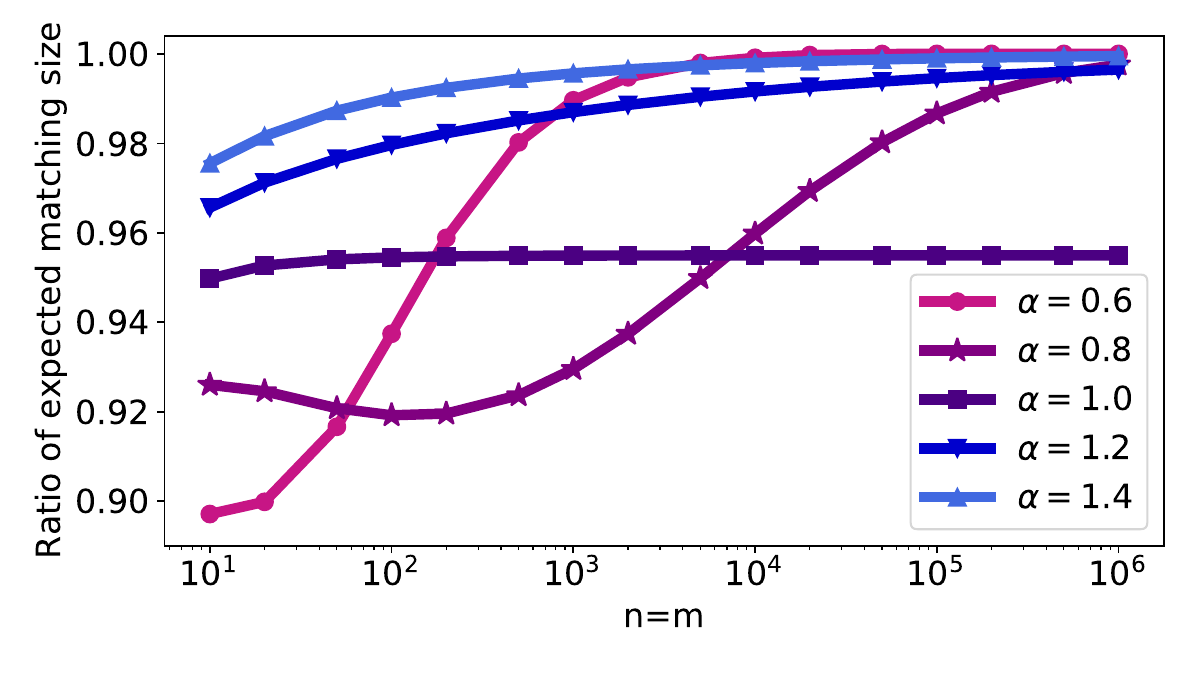}}
\caption{Ratio of MPD's expected matching size and the expected maximum matching size on symmetric \GraphModel{} random graphs with offline expected degrees following Zipf's Law with exponent $\alpha$. The $i$th offline node has expected degree $d_i = C i^{-\alpha}$ with $C = m/2$.}
\label{fig-analysis-zipf}
\end{center}
\end{figure}

Using the equations derived in Section~\ref{sec-analysis}, in Figure~\ref{fig-analysis-zipf} we plot the competitive ratio for symmetric \GraphModel{} graphs with expected offline degrees following Zipf's Law (at $n=m=1000$, this setup corresponds to the experiment shown in Figure~\ref{fig-zipf}).
Across all choices of the exponent other than $\alpha=1$, MPD's performance relative to the size of the maximum matching increases as the size of the graph grows.
Further, for larger graphs, in many settings MPD achieves a ratio close to $1$ and even for smaller graph achieves ratios above $0.9$ for almost all settings.

Notably, at $\alpha=1$, the ratio plateaus around $0.955$, and at $\alpha=0.8$, the ratio decreases from $n=10$ to $n=100$ before rising again as $n$ increases.
At $\alpha = 1$ across all values of $n$ as well as at $\alpha=0.8$ with $n=100$, a large fraction of the offline nodes have expected degree close to one.
Many of these nodes will have actual degree $1$ and many will have actual degree $\geq 2$.
MPD has no way of distinguishing between these two types of nodes as it only uses expected degrees and will mistakenly not match some offline nodes that only appear once.
While there is always a discrepancy between actual and expected degrees, the issue of prioritizing a node with actual degree $\geq 2$ over a node with actual degree $1$ is most detrimental, leading to worse performance when there are many offline nodes with expected degree close to one.

\section{Solution to System of Differential Equations in Equation~\ref{eq-diffeq}}
\label{appendix-diffeqsoln}
Recall the system of differential equations from Equation~\ref{eq-diffeq}:
\begin{equation}
    \frac{dz_d(t)}{dt} = k_d \left(1 - e^{z_d(t)}\right) \prod_{d' < d} e^{z_{d'}(t)}
\end{equation}
for all unique expected degrees $d$ in $\mathbf{d}$.

The solution to the system of differential equations is given by the following equations.
Let $\{\delta_i\}_{i=1}^\ell$ be the ordered set of unique expected degrees and let $f_d$ be the number of offline nodes with expected degree $d$.
We will define the auxiliary functions $\alpha_{\delta_i}(t)$ for $i \in \{2,\ldots,\ell\}$ and variables $C_{\delta_i}$ for $i=\{1,\ldots,\ell\}$ as follows:
\begin{align*}
    \alpha_{\delta_2}(t) &= C_{\delta_1} + e^{k_{\delta_1} t}\\
    \alpha_{\delta_i}(t) &= (\alpha_{\delta_{i-1}}(t))^{k_{\delta_{i-1}}/k_{\delta_{i-2}}} + C_{\delta_{i-1}} \text{ (for $i \geq 3$)} 
\end{align*}
where 
\begin{align*}
    C_{\delta_1} &= e^{k_{\delta_1} f_{\delta_1}} - 1 \\
    C_{\delta_i} &= (\alpha_{\delta_i}(0))^{k_{\delta_i}/k_{\delta_{i-1}}} (e^{k_{\delta_i} f_{\delta_i}} - 1) \text{ (for $i \geq 2$)}.
\end{align*}
Then,
\begin{equation}
\label{eq-diffeqsoln}
\begin{split}
    z_{\delta_1}(t) &= -\log(C_{\delta_1} e^{-k_{\delta_1} t} + 1) \\
    z_{\delta_i}(t) &= -\log(C_{\delta_i} (\alpha_{\delta_i}(t))^{-k_{\delta_i}/k_{\delta_{i-1}}} + 1) \text{ (for $i \geq 2$)}.
\end{split}
\end{equation}

In the rest of this section, we show that Equation~\ref{eq-diffeqsoln} give the correct solutions to the differential equations in Equation~\ref{eq-diffeq}.

\begin{lemma}
\label{lemma-diffeq}
For $i \in \{2,3,\ldots, \ell\}$,
\begin{equation}
    \frac{\frac{d}{dt} \left(\alpha_{\delta_i}(t)\right)}{k_{\delta_{i-1}} \alpha_{\delta_i}(t)}  = \prod_{j=1}^{i-1} e^{z_{\delta_j}(t)}.
\end{equation}
\end{lemma}

\begin{proof}
We will prove the lemma by induction.
Consider the base case of $i = 2$:
\begin{equation*}
\begin{split}
     \frac{\frac{d}{dt} \left(\alpha_{\delta_2}(t)\right)}{k_{\delta_1} \alpha_{\delta_2}(t)} &= \frac{k_{\delta_1} e^{k_{\delta_1} t}}{k_{\delta_1} (C_{\delta_1} + e^{k_{\delta_1} t})} \\
     &= \frac{1}{C_{\delta_1} e^{-k_{\delta_1} t} + 1} \\
     &= e^{z_{\delta_1}(t)}.
\end{split}
\end{equation*}

Now consider the inductive case of $i > 2$ under the assumption that $\frac{\frac{d}{dt}(\alpha_{\delta_{i-1}})}{k_{\delta_{i-2}} \alpha_{\delta_{i-1}}} = \prod_{j=1}^{i-2} e^{z_{\delta_j}}$:
\begin{equation*}
\begin{split}
     \frac{\frac{d}{dt} \left(\alpha_{\delta_i}(t)\right)}{k_{\delta_{i-1}} \alpha_{\delta_i}(t)} &= \frac{\frac{k_{\delta_{i-1}}}{k_{\delta_{i-2}}} \alpha_{\delta_{i-1}}^{(k_{\delta_{i-1}}/k_{\delta_{i-2}})-1} \frac{d}{dt} (\alpha_{\delta_{i-1}})}{k_{\delta_{i-1}} (\alpha_{\delta_i-1}^{k_{\delta_{i-1}}/k_{\delta_{i-2}}} + C_{\delta_{i-1}})} \\
     &= \frac{\alpha_{\delta_{i-1}}^{k_{\delta_{i-1}}/k_{\delta_{i-2}}}}{\alpha_{\delta_i-1}^{k_{\delta_{i-1}}/k_{\delta_{i-2}}} + C_{\delta_{i-1}}} \cdot \frac{\frac{d}{dt}(\alpha_{\delta_{i-1}})}{k_{\delta_{i-2}} \alpha_{\delta_{i-1}}} \\
     &= \frac{1}{C_{\delta_{i-1}} \alpha_{\delta_{i-1}}^{-k_{\delta_{i-1}}/k_{\delta_{i-2}}} + 1} \cdot \prod_{j=1}^{i-2} e^{z_{\delta_j}(t)} \\
     &= \prod_{j=1}^{i-1} e^{z_{\delta_j}(t)}.
\end{split}
\end{equation*}
This completes the proof.
\end{proof}

\begin{lemma}
The expressions for $z_{\delta_i}(t)$ in Equation~\ref{eq-diffeqsoln} give a solution to system of differential equations in Equation~\ref{eq-diffeq} with initial conditions $z_{\delta_i}(0) = -k_{\delta_1} f_{\delta_i}$.
\end{lemma}
\begin{proof}
We will split the proof into two cases for $\delta_1$ and for $\delta_i$ with $i \geq 2$.
Starting with $i=1$, recall that 
\[
    z_{\delta_1}(t) =-\log(C_{\delta_1} e^{-k_{\delta_1} t} + 1).
\]
First, we will show that this function has the correct derivative.
\begin{equation*}
\begin{split}
    \frac{dz_{\delta_1}(t)}{dt} &= - \frac{1}{C_{\delta_1} e^{-k_{\delta_1} t} + 1} (C_{\delta_1}) (-k_{\delta_1}) e^{-k \delta_1 t} \\
    &= k_{\delta_1} \frac{C_{\delta_1} e^{-k_{\delta_1} t}}{C_{\delta_1} e^{-k_{\delta_1} t} + 1} \\
    &= k_{\delta_1} \left(1 - e^{z_{\delta_1}(t)}\right)
\end{split}
\end{equation*}
It remains to be shown that $z_{\delta_1}(0) = - k_{\delta_1} f_{\delta_1}$:
\begin{equation*}
\begin{split}
    z_{\delta_1}(0) &=-\log(C_{\delta_1} + 1) \\
    &= -\log(e^{k_{\delta_1} f_{\delta_1}}) \\
    &= - k_{\delta_1} f_{\delta_1}.
\end{split}    
\end{equation*}

Now, consider the case where $i \geq 2$.
Then,
\begin{equation*}
\begin{split}
    \frac{dz_{\delta_i}(t)}{dt} &= \frac{d}{dt}\left(-\log(C_{\delta_i} (\alpha_{\delta_i}(t))^{-k_{\delta_i}/k_{\delta_{i-1}}} + 1)\right) \\
    &= -\frac{C_{\delta_i} \frac{-k_{\delta_i}}{k_{\delta_{i-1}}} (\alpha_{\delta_i}(t))^{-k_{\delta_i}/k_{\delta_{i-1}}-1}}{C_{\delta_i} (\alpha_{\delta_i}(t))^{-k_{\delta_i}/k_{\delta_{i-1}}} + 1} \frac{d}{dt} \left(\alpha_{\delta_i}(t)\right) \\
    &= k_{\delta_i} \frac{C_{\delta_i} (\alpha_{\delta_i}(t))^{-k_{\delta_i}/k_{\delta_{i-1}}}}{C_{\delta_i} (\alpha_{\delta_i}(t))^{-k_{\delta_i}/k_{\delta_{i-1}}} + 1} \frac{\frac{d}{dt} \left(\alpha_{\delta_i}(t)\right)}{k_{\delta_{i-1}} \alpha_{\delta_i}(t)} \\
    &= k_{\delta_i} \left(1 - e^{z_{\delta_i}(t)}\right) \prod_{j=1}^{i-1} e^{z_{\delta_j(t)}}.
\end{split}
\end{equation*}
The last step makes use of Lemma~\ref{lemma-diffeq}.
Finally, we must show that $z_{\delta_i}(0) = -k_{\delta_i} f_{\delta_i}$:
\begin{equation*}
\begin{split}
    z_{\delta_i}(0) &= -\log(C_{\delta_i} \alpha_{\delta_i}(0)^{-k_{\delta_i}/k_{\delta_{i-1}}} + 1) \\
    &= -\log(e^{k_{\delta_i} f_{\delta_i}} - 1 + 1) \\
    &= -k_{\delta_i} f_{\delta_i}.
\end{split}
\end{equation*}
Thus, the given solution to the system of differential equations is correct.
\end{proof}

\section{Proof of Theorem~\ref{thm-mpd-expectation}}
\label{appendix-mpd-expectation}
We give the proof of Theorem~\ref{thm-mpd-expectation} which states that the solution to the differential equations models the size of the matching returned by MinPredictedDegree.
\begin{proof}[Proof of Theorem~\ref{thm-mpd-expectation}]
The proof of Theorem~\ref{thm-mpd-expectation} follows a direct application of Theorem 1 in Luby et al.~\cite{luby2001efficient}.
We must show three conditions are satisfied: (i) that $|Z_d^{t+1} - Z_d^t|$ is bounded, (ii) that $\frac{dz_d(t)}{dt} = \E[Z_d^{t+1} - Z_d^t | H_t]$ (where $H_t$ is the history up to time $t$, and (iii) that $\frac{dz_d(t)}{dt}$ satisfies a Lipschitz condition when $z_d(t) \leq 0$ (recall that $Z_d^t$ is always nonpositive as $Y_d^t$ is always nonnegative).

If these conditions hold, then the solution to the system of differential equations gives the asymptotic expected behavior of the variables $Z_d^t$ and thus the asymptotic expected behavior of the variables $Y_d^t$, which govern the size of the matching returned by MinPredictedDegree.
% In addition, we directly get the desired concentration bound which tells us in the non-asymptotic case that $Y_d^t$ does not deviate too far from the predicted behavior from the differential equations:
% \begin{equation}
%     \Pr(Y_d > -z_{\delta_i}(m)/k + cm^{5/6}) < \ell m^{2/3} \exp(- m^{1/3}/2).
% \end{equation}

It remains to show that these three conditions are met.
Condition (i) is satisfied as the number of nodes of a given expected degree can change by at most one per timestep, so $|Z_d^{t+1} - Z_d^t| \leq k_d$.
Condition (ii) is satisfied by construction in Equation~\ref{eq-diffeq}.
Finally, Condition (iii) is satisfied as $\frac{dz_d(t)}{dt}$ is comprised of a product of several terms resembling $C_1 \cdot e^{-C_2 x}$ for nonnegative constants $C_1, C_2$ and with $x$ nonnegative. Therefore, $\frac{dz_d(t)}{dt}$ has constant bounded first derivatives when $z_d(t) \leq 0$.
\end{proof}

\section{Upper Bound on Expected Maximum Matching Size}
\label{appendix-offline-expectation}
\subsection{Overview}
To analyze the maximum matching size within this model, we rely on a upper bound based on the matching version of Hall's marriage theorem~\cite{hall1935onrepresentatives}.
We first state the classic theorem.
\begin{theorem}[Hall's Theorem]
Let $G=(U \cup V, E)$ be a bipartite graph.
For any subset of nodes $S$, let $N(S)$ be the set of neighbors of the nodes in $X$.
Then, $G$ has a perfect matching if and only if for all $S \subset U$ and $T \subset V$, $|S| \leq |N(S)|$ and $|T| \leq |N(T)|$.
\end{theorem}
Intuitively, if there is any subset $S$ with few neighbors, then only $|N(S)|$ of the members of $S$ can possibly be matched.
Let $\mu(G)$ correspond to the size of the maximum matching in $G$.
For any bipartite graph $G=(U \cup V, E)$ and $S \subset U$,
\begin{equation}
\label{eq-offline-upper-bound}
    \mu(G) \leq n - (|S| - |N(S)|)
\end{equation}
as out of all of the nodes in $S$, only $|N(S)|$ can be matched.
Therefore, if we can calculate the expected size of $|S|$ and $|N(S)|$ for some subset of a random symmetric \GraphModel{} graph $G$, we immediately get an upper bound on the expected size of the maximum matching in $G$.

Our upper bound involves constructing a specific subset $S^*$ of the offline nodes that makes use of our focus on power law graphs to provide a useful bound that is easy to evaluate.
We empirically test how good of a bound Equation~\ref{eq-offline-upper-bound} gives using $S^*$ and find that for symmetric \GraphModel{} random graphs with offline degrees following a power law distribution, the upper bound on $\mu(G)$ given by $n - (|S^*| - |N(S^*)|)$ is close to maximum matching size (often achieving the same value and in all trials was less than $2\%$ greater than the true value).
In addition to providing a good bound, we show that we can evaluate $\E[|S|]$ and $\E[|N(S)|]$ in Equations~\ref{eq-offline-ns},~\ref{eq-offline-s},~\ref{eq-offline-s2}.
By linearity of expectation, this directly gives us an upper bound on the expected size of the maximum matching.

In Appendix~\ref{appendix-offline-concentration}, we show that on symmetric \GraphModel{} random graphs with power law distributed degrees, our bound on the maximum matching size is concentrated about its expectation.
Combined with Theorem~\ref{thm-mpd-concentration}, this implies that our analytic results on the ratio of the expected sizes in Section~\ref{sec-analysis} are closely related to the competitive ratio.

\subsection{Construction}
Let $S^*$ be the subset of $U$ constructed as follows.
Let $U_1$ be the set of degree 1 nodes in $U$ (here degree 1 referring to the actual degree of the node rather than the expected degree in the \GraphModel{} model).
Then, $S^*$ is the maximal set of nodes in $U$ s.t. $N(S^*) \subseteq N(U_1)$.
In other words, $S^*$ is the maximal subset of nodes in $U$ whose neighbors completely overlap with the neighbors of the degree 1 nodes of $U$.

% From Equation~\ref{eq-offline-upper-bound}, the size of the sets $S^*$ and $N(S^*)$ can be used to upper bound the the maximum matching size. Empirically, we find that for \GraphModel{} random graphs with offline degrees following a power law distribution and with $n=m$, the upper bound on $\mu(G)$ given by $n - (|S^*| - |N(S^*)|)$ is close to maximum matching size (often achieving the same value and in all trials was less than $2\%$ greater than the true value).

As $n - (|S^*| - |N(S^*)|)$ gives an upper bound on the maximum matching size, $\E[n - (|S^*| - |N(S^*)|)]$ gives an upper bound on the expected maximum matching size.
The expected sizes of $S^*$ and $N(S^*)$ in a \GraphModel{} random graph with expected degrees $\mathbf{d}$ are given by the following equations.
The expected size of $N(S^*)$ is simply the sum over all online nodes $v \in V$ of the probability that $v$ has at least one degree $1$ neighbor.
The expected size of $S^*$ is broken down as the sum over all offline nodes $u \in U$ of the probability that $u$ has actual degree $\Delta$ and then the probability that all $\Delta$ of $u$'s neighbors are members of $N(S^*)$.
Let $S^*_\Delta$ be the subset of nodes in $S^*$ whose actual (as opposed to expected) degree are $\Delta$ and let $\beta^\Delta_i$ for $\Delta \in \{0,...,m\}$ and $i \in \{1,...,n\}$ be defined as
\begin{equation*}
   \beta^0_i = \beta^1_i = 1 \\ 
\end{equation*}
\begin{equation*}
\begin{split}
   \beta^\Delta_i &= 1 + \sum_{r=1}^\Delta (-1)^r \binom{\Delta}{r} \prod_{i'\neq i}\left[1 - r \left(\frac{d_{i'}}{m}\right) \left(1 - \frac{d_{i'}}{m}\right)^{m-1}\right]\\
   &(\text{for } \Delta \geq 2).
\end{split}
\end{equation*}
$\beta_i^\Delta$ represents the probability of $u_i \in S$ conditioned on $u_i$ having actual degree $\Delta$.
Then,
\begin{equation}
\label{eq-offline-ns}
    \E[|N(S^*)|] = m\left(1 -  \prod_{i=1}^n \left[1 - \frac{d_i}{m} \left(1 - \frac{d_i}{m}\right)^{m-1}\right] \right)
\end{equation}
and
\begin{equation}
\label{eq-offline-s}
    \E[|S^*|] = \sum_{\Delta=0}^m \E[|S^*_\Delta|] 
\end{equation}
where
\begin{equation}
\begin{split}
\label{eq-offline-s2}
    \E[|S^*_\Delta|] &= \sum_{i=1}^n \binom{m}{\Delta} \left(\frac{d_i}{m}\right)^\Delta \left(1 - \frac{d_i}{m}\right)^{m - \Delta} \beta^\Delta_i.
\end{split}
\end{equation}

In the rest of this section, we show that the equations for the expected size of $|S^*|$ and $|N(S^*)|$ are correct by showing their derivations.

First, consider $\E[|N(S^*)|]$.
Recall that $N(S^*)$ is the set of online nodes that have a neighbor with actual degree $1$.
Therefore, the expected size of $N(S^*)$ is $m$ minus the expected number of online nodes that have no degree one neighbors.
For any online node $v \in V$, the probability that $v$ has no degree 1 neighbors is
\begin{equation*}
\begin{split}
    \prod_{u \in U} &\Pr(\text{$u$ is not a deg 1 nbr of $v$}) \\
    &= \prod_{u \in U} [1 - \Pr(\text{$u$ nbr of $v$}) \Pr(\text{$u$ has no other nbrs})] \\
    &= \prod_{u \in U} \left[1 - \frac{d_u}{m} \left(1 - \frac{d_u}{m}\right)^{m-1}\right].
\end{split}
\end{equation*}
By linearity of expectation,
\[
    \E[|N(S)|] = m\left(1 - \prod_{u \in U} \left[1 - \frac{d_u}{m} \left(1 - \frac{d_u}{m}\right)^{m-1}\right]\right).
\]

Now, we will deal with $\E[|S^*_\Delta|]$. Recall that $S^*_\Delta$ is the set of offline nodes with actual degree $\Delta$ with all of their online neighbors having at least one offline neighbor with actual degree 1.
For a given offline node $u \in U$ with expected degree $d_u$, the probability of $u$ being in $S^*_\Delta$ is the product of the probability of $u$ having actual degree $\Delta$ and the conditional probability of all of $u$'s neighbors having a degree 1 neighbor given $u$ having actual degree $\Delta$.
We will call the first event $A_{u, \Delta}$ and the second, conditional event $B_{u, \Delta} | A_{u, \Delta}$.
The probability of $A_{u, \Delta}$ occurring corresponds to a Binomial random variable with size parameter $m$ and probability parameter $d_u/m$ taking on value $\Delta$:
\begin{equation*}
\begin{split}
    \Pr(A_{u, \Delta}) = \binom{m}{\Delta} \left(\frac{d_u}{m}\right)^\Delta  \left(1 - \frac{d_u}{m}\right)^{m - \Delta}.
\end{split}
\end{equation*}
The probability of $B_{u, \Delta} | A_{u, \Delta}$ equals 1 if $\Delta = 0$ or $\Delta = 1$ as either $u$ has no neighbors or $u$ is itself a degree 1 neighbor of its neighbors, respectively.
If $\Delta \geq 2$, then $\Pr(B_{u, \Delta} | A_{u, \Delta})$ is equal to the complement of the event that at least one of $u$'s neighbors has no degree 1 neighbor.
Let $C_{u, \Delta, r}|A_{u, \Delta}$ be the event that any subset of $r$ of $u$'s $\Delta$ neighbors have no degree 1 neighbor given $A_{u, \Delta}$.
$\Pr(C_{u, \Delta, r}|A_{u, \Delta})$ can be expressed as
\[
    \binom{\Delta}{r} \prod_{u' \neq u} \left[1 - r\left(\frac{d_{u'}}{m}\right)\left(1 - \frac{d_{u'}}{m}\right)^{m-1}\right]
\]
where the term within the product represents the probability of no offline nodes (excluding $u$) being degree one neighbors of a specific set of $r$ online nodes (similarly to when expressing $\E[|N(S^*)|]$ above).
By the inclusion-exclusion rule,
\[
    \Pr(B_{u, \Delta} | A_{u, \Delta}) = 1 + \sum_{r=1}^\Delta (-1)^r \Pr(C_{u, \Delta, r}|A_{u, \Delta}),
\]
thus completing the derivation of Equations~\ref{eq-offline-ns},~\ref{eq-offline-s}, and~\ref{eq-offline-s2}.

\section{Analysis on \GraphModel{} random graphs in asymptotic case}
\label{appendix-asymptotic}
In this section, we give slight modifications of the Equation~\ref{eq-diffeqsoln} and Equations~\ref{eq-offline-ns},~\ref{eq-offline-s},~\ref{eq-offline-s2} in the case where $n, m \rightarrow \infty$ to allow us to evaluate the equations to produce the results in Table~\ref{table-analysis-expcut}.
The model will change slightly when considering the asymptotic case: we will describe the set of offline expected degrees $\mathbf{d}$ by a set of unique degrees $\{\delta_i\}_{i=1}^\ell$ and corresponding fractions $\{\lambda_i\}_{i=1}^\ell$ where a $\lambda_i$ fraction of the offline nodes have expected degree $\delta_i$.

Importantly for the asymptotic results in Table~\ref{table-analysis-expcut}, while there are offline nodes with expected degree approaching infinity, a finite number of unique expected degrees account for all but an exponentially small fraction of the offline nodes, allowing us to evaluate the equations up to negligible error.
In the following calculations, we will consider both $\ell$ as well as all $\delta_i$ for $i =\{1,...,\ell\}$ to be finite.

\subsection{Asymptotic analysis of MinPredictedDegree}
To start, we will replace $f_d$ with $m \cdot \lambda_d$ and we will replace $t$ with $\tau = t/m$.
Recall $k_d = -\log(1 - d/m)$.
The Taylor expansion of $\log(1-x)$ at $x=0$ is $-\sum_{n=1}^{\infty} \frac{x^n}{n}$.
Within Equation~\ref{eq-diffeqsoln}, $k_d$ appears in terms $k_d/k_{d'}$ and $k_d \cdot f_d = k_d \cdot m \cdot \lambda_d$.
In the asymptotic case, we will use the following substitutions for those terms:
\begin{equation*}
    \lim_{m \rightarrow \infty} k_d/k_{d'} = d/d'.
\end{equation*}
and
\begin{equation*}
    \lim_{m \rightarrow \infty} k_d \cdot m \cdot \lambda_d = -d \cdot \lambda_d.
\end{equation*}
In both cases, we use the fact that as $m \rightarrow \infty$, the first term in the Taylor series $(d/m)$ dominates.

Using these substitutions, we can rewrite the equations for MinPredictedDegree as follows.
\begin{align*}
    \alpha_{\delta_2}(\tau) &= C_{\delta_1} + e^{\delta_1 \tau}\\
    \alpha_{\delta_i}(\tau) &= (\alpha_{\delta_{i-1}}(\tau))^{\delta_{i-1}/\delta_{i-2}} + C_{\delta_{i-1}} \text{ (for $i \geq 3$)} 
\end{align*}
where 
\begin{align*}
    C_{\delta_1} &= e^{\delta_1 \lambda_{\delta_1}} - 1 \\
    C_{\delta_i} &= (\alpha_{\delta_i}(0))^{\delta_i/\delta_{i-1}} (e^{\delta_i \lambda_{\delta_i}} - 1) \text{ (for $i \geq 2$)}.
\end{align*}
Then,
\begin{equation}
\label{eq-diffeqsoln-asymptotic}
\begin{split}
    z_{\delta_1}(\tau) &= -\log(C_{\delta_1} e^{-\delta_1 \tau} + 1) \\
    z_{\delta_i}(\tau) &= -\log(C_{\delta_i} (\alpha_{\delta_i}(\tau))^{-\delta_i/\delta_{i-1}} + 1) \text{ (for $i \geq 2$)}.
\end{split}
\end{equation}

Recall that the expected number of offline nodes with expected degree $d$ is given by $-z_d(\tau)/k_d$ evaluated when $t=m$.
Then, in the asymptotic case, the expected fraction of offline nodes matched is 
\begin{equation}
\label{eq-mpd-expectation-asymptotic}
    \sum_{i=1}^\ell \lambda_i + z_{\delta_i}(1)/\delta_i
\end{equation}
where $z_{\delta_i}(\tau)$ are given by Equation~\ref{eq-diffeqsoln-asymptotic}.

\subsection{Asymptotic analysis of maximum matching}
For the equations for the upper bound on the expected maximum matching size, the key fact we will use is $\lim_{x \rightarrow 0} (1+x) = e^x$.
Therefore, we can replace all terms of $(1 - \frac{d}{m})^{m-C}$ with $e^{-d}$.
In addition, we can replace all terms of $\binom{m}{C} \left(\frac{d}{m}\right)^C$ with $\frac{d^C}{C!}$.
These substitutions give the following equations.
\begin{equation*}
   \beta^0 = \beta^1 = 1 \\ 
\end{equation*}
\begin{equation*}
\begin{split}
   \beta^\Delta &= 1 - \sum_{r=1}^\Delta (-1)^r \binom{\Delta}{r} \prod_{i=1}^\ell \left[1 - r \left(\frac{\delta_i}{m}\right) e^{-\delta_i}\right]^{m*\lambda_i}\\
   &= 1 - \sum_{r=1}^\Delta (-1)^r \binom{\Delta}{r} \prod_{i=1}^\ell e^{-r \delta_i \lambda_i e^{-\delta_i}} (\text{ for } \Delta \geq 2).
\end{split}
\end{equation*}
Note that the $\beta^\Delta$ terms are no longer indexed by $i$ as conditioning on the actual degree of a single node makes no difference on the probability in the asymptotic case.
Then,
\begin{equation}
    \E\left[\frac{|N(S^*)|}{m}\right] = \left(1 -  \prod_{i=1}^\ell e^{-\delta_i \lambda_i e^{-\delta_i}} \right)
\end{equation}
and
\begin{equation}
    \E\left[\frac{|S^*|}{m}\right] \geq \sum_{\Delta=0}^C \E\left[\frac{|S^*_\Delta|}{m}\right] 
\end{equation}
where
\begin{equation}
\begin{split}
    \E\left[\frac{|S^*_\Delta|}{m}\right] &= \sum_{i=1}^\ell \lambda_i \frac{(\delta_i)^\Delta}{\Delta!} e^{-\delta_i} \beta^\Delta_i.
\end{split}
\end{equation}

% \color{blue}
\section{Analysis of MPD for Erd\H{o}s-Rényi Random Bipartite Graphs} \label{appendix-uniform}
We here analyze the performance of MPD on CLV-B instances, $I_{\textbf{p},\textbf{1}}$, where $\textbf{p}=(p,\dots,p) \in [0,1]^n$ and $\textbf{1}=(1,\dots,1) \in [0,1]^m$. Such an instance is an Erd\H{o}s-Rényi bipartite random graph where all edges appear with the same probability $p$. In particular, the expected degrees of the offline vertices are the same and equal to $d:=mp$. Note that in this case, MPD is equivalent to any other greedy algorithm. Letting $c=m/n$, we show that for a wide range of the parameters $m,n$, and $p$, the expected fraction of matched offline vertices is $1+c-\frac{c\ln(e^{d}+e^{d/c}-1)}{d}$ up to a small additive error.

Combining this bound, with the asymptotic upper bound on the maximum matching of Appendix~\ref{appendix-asymptotic}, we obtain that for these graphs, the asymptotic competitive ratio of MPD is at least 0.831 which is significantly better than the $0.7299$ bound from~\cite{brubach2016new}. 
We conjecture that Erd\H{o}s-Rényi random bipartite graphs are in fact worst-case instances for MPD, in the sense that a lower bound on the competitive ratio of MPD for Erd\H{o}s-Rényi random graphs, also holds for general CLV-B random graphs. 

%We will assume that $m=\Theta(n)$ and moreover, that $d=o(\log n)$. 
%If $d=\omega(1)$, it is easy to check that the size of the matching output by any greedy algorithm is $(1-o(1))\min(n,m)$ with high probability, and as a maximum matching has size at most $\min(n,m)$, the algorithm is $(1-o(1))$-competitive.
\begin{theorem}\label{thm:mpd_uniform}
Let $p\in [0,1]$ and $n,m\in \N$. Assume that\footnote{Here, $o(1)\to 0$ as $n\to \infty$} $n^{1-o(1)}\leq m\leq n^{1+o(1)}$, $p=o(\log n)/m$, and $p\geq 1/n^{1+o(1)}$. Let $\textbf{p}=(p,\dots,p)\in [0,1]^n$, and $\textbf{q}=(1,\dots,1)\in [0,1]^m$. Let $M$ be the size of the matching output by any greedy algorithm $\mathcal{A}$ on input $I_{\textbf{p},\textbf{q}}$. Let $c=m/n$. Then 
$$
\frac{\E[M]}{n}=1+c-\frac{c\ln(e^{d}+e^{d/c}-1)}{d}\pm n^{-1/2+o(1)}.
$$
Moreover, $|M-\E[M]|=O(\sqrt{m\log n})$ with high probability in $n$.
\end{theorem}
\begin{proof}
Let $\mathcal{F}_i$ denote the $\sigma$-algebra generated by the neighborhoods of the first $i$ arriving online nodes. Defining $X_i=\E[M\mid \mathcal{F}_i]$, we  have that $(X_i)_{i=0}^m$ is a martingale with $X_0=\E[M]$ and $X_n=M$. It is easy to check that $|X_i-X_{i-1}|\leq 1$ for $1\leq i \leq m$, so it follows from Azuma's inequality that for any $t>0$,
$$
\Pr[|M-\E[M]|\geq t]\leq \exp\left( \frac{-t^2}{2m} \right)
$$
In particular, $|M-\E[M]|= O(\sqrt{m\log n})$ with high probability, say at least $1-n^{-10}$ as claimed in the theorem. 
%Moreover, the assumptions that $m=\Theta(n)$ and $d=O(1)$ implies that $\E[M]=(1-\Omega(1))n$
Moreover, the assumption $pm=o(\log n)$ implies that $\E[M]\leq n(1-n^{-o(1)})$. This is clear when $m\leq n/2$. On the other hand when $m>n/2$, the probability that any given offline node is never picked by an online node is $(1-p)^m=e^{-o(\log n)}=n^{-o(1)}$, using the assumption $pm=o(\log n )$.

For $0\leq j\leq n-1$, we let $T_j$ denote the number of online vertices $v$ such that when $v$ arrives, the matching found by $\mathcal{A}$ so far has size $j$. Then $T_j$ is geometrically distributed with parameter $1-(1-p)^{n-j}$, and so, $\E[T_j]=\frac{1}{1-(1-p)^{n-j}}$. Moreover, $T_j= O(\frac{\log n}{1-(1-p)^{n-j}})$ with high probability.

Define $n^*=\lfloor\E[M]\rfloor$. Let $L_1=n(1-n^{-o(1)})$ be such that $\max(M,\E[M])\leq L_1$ with probability at least $1-n^{-10}$ and let $A_1$ be the event that $\max(M,\E[M])> L_1$. Note that when $j\leq L_1$, then 
$$\E[T_j]\leq \frac{1}{1-(1-p)^{n^{1-o(1)}}}\leq \frac{1}{1-\exp(-pn^{1-o(1)})}=\frac{1}{1-\exp(-n^{-o(1)})}=n^{o(1)}.
$$
We can therefore pick $L_2=n^{o(1)}$ such that $\max(T_1,\dots,T_{L_1})\leq L_2$ with probability at least $1-n^{-10}$ and we let $A_2$ be the event that $\max(T_1,\dots,T_{L_1})> L_2$. Finally, let $L_3=O(\sqrt{m\log n})$ be such that $|M-\E[M]|\leq L_3$ with probability at least $1-n^{-10}$ and let $A_3$ be the event that $|M-\E[M]|> L_3$

If neither $A_1$, $A_2$, or $A_3$ occur, which happens with probability at least $1-3n^{-10}$, then
$$
\left|m-\sum_{i\leq n^*}T_i \right|\leq \left|m-\sum_{i\leq M}T_i \right|+L_2L_3\leq L_2+L_2L_3=n^{1/2+o(1)}.
$$
From this, it particularly follows that 
\begin{align}\label{eq:concentration-juggle}
\left|m-\E\sum_{i\leq n^*}T_i \right|\leq n^{1/2+o(1)}.
\end{align}
Let $n_0$ be minimal such that $\sum_{j<n_0} \E[T_j]\geq m$. Since $T_i\geq 1$ for every $i$, it follows from~\eqref{eq:concentration-juggle} that $|n_0-n^*|\leq n^{1/2+o(1)}$, and in particular that $|n_0-\E[M]|\leq n^{1/2+o(1)}$. To finish the proof, it therefore suffices to show that $n_0/n$ satisfies the bound in the theorem. 

For this, we first note that
\begin{align*}
    m\leq \sum_{j<n_0}\E[T_j]
    =&\sum_{j<n_0}\frac{1}{1-(1-p)^{n-j}}
    \leq \sum_{j<n_0} \frac{1}{1-e^{-p(n-j)}}
    \leq \int_0^{n_0}\frac{1}{1-e^{-p(n-x)}} \, dx. \\
    =& n_0-\frac{1}{p}\left(\ln(1-e^{p(n_0-n)})-\ln(1-e^{-pn}) \right)
\end{align*}
The right hand side is an increasing function of $n_0$ vanishing at $0$ and turning to infinity as $n_0 \to n$. Solving for $n_0$ then gives that
\[
  n_0\geq n+m-\frac{\ln(e^{pm}+e^{pn}-1)}{p},
\]
so that
\[
    n_0/n\geq 1+m/n-\frac{\ln(e^{pm}+e^{pn}-1)}{pn}=1+c-\frac{c\ln(e^{d}+e^{d/c}-1)}{d}
\]
which yields the lower bound in the proof of the theorem upon dividing by $n$.

For the upper bound, we use the inequality $(1+\frac{x}{n})^n\geq e^x(1-\frac{x^2}{n})$ holding for $n\geq 1$ and $|x|\leq n$, from which it follows that 
$$
\E[T_j]\geq \frac{1}{1-e^{-p(n-j)}(1-(n-j)p^2)}
\geq \frac{1}{1-e^{-p(n-j)}a},
$$
where we have put $a=1-np^2$. Note that by the definition of $n_0$, $\sum_{j< n_0-1} \E[T_j]<m$, and so 
it follows similarly to above that
\begin{align*}
   m_0&:=m+\E[T_{n_0-1}+T_{n_0}]\geq \sum_{j\leq n_0} \E[T_j]
    \geq \int_0^{n_0}\frac{1}{1-e^{-p(n-x)}a} \, dx  \\
    &\geq n_0-\frac{1}{p}\left(\ln(1-ae^{p(n_0-n)})-\ln(1-ae^{-pn}) \right).
\end{align*}
The right hand side is again an increasing function of $n_0$, so solving for $n_0$ gives that 
\begin{align*}
n_0\leq & m_0+n-\frac{\ln(e^{pn}+ae^{pm}-a)}{p}=m_0+n-\frac{\ln(e^{pn}+e^{pm}-1-p^2n(e^{pm}-1))}{p}\\
=&m_0+n-\frac{\ln(e^{pn}+e^{pm}-1)}{p}+O(pne^{pm})=m+n-\frac{\ln(e^{pn}+e^{pm}-1)}{p}+2n^{o(1)}.
\end{align*}
The desired result follows after dividing by $n$.
\end{proof}
\paragraph{Competitive ratio of MPD with uniform expected degree sequence.}
We now demonstrate how to combine the bound of Theorem~\ref{thm:mpd_uniform} with the upper bounds on the maximum matching of Appendix~\ref{appendix-asymptotic}, to obtain the better bound of 0.831 on the asymptotic competitive ratio of MPD for Erd\H{o}s-Rényi bipartite random graph. We start with the following lemma that allows us to focus on the case where both $c=m/n$ and $d=pm$ are constants, $d,c=\Theta(1)$.
\begin{lemma}\label{lem:assume-constants}
There exists an $\eps>0$, so that if $(c,d)\in [\eps,1/\eps]^2$, and $n$ is sufficiently large, then the competitive ratio of MPD on instance $I_{\textbf{p},\textbf{q}}$ is at least $0.99$. Here, $\textbf{p}=(p,\dots,p)\in [0,1]^n$, and $\textbf{q}=(1,\dots,1)\in [0,1]^m$.
\end{lemma}
\begin{proof}[Proof (sketch)]
Let $S_1$ denote the set of non-isolated offline nodes, $S_2$ the set of non-isolated online nodes, $M^*$ the maximum matching, and $M$ the matching found by MPD. First of all, it is easy to check that there exists $\eps_1>0$, so that if  if $c<\eps_1$ and $n$ is large enough, then both $|1-\frac{\E[|M^*|]}{\E[|S_2|]}|<\frac{1}{1000}$ and $|1-\frac{\E[|M|]}{\E[|S_2|]}|<\frac{1}{1000}$. The first bound follows from observing that for $c$ small, nearly every non-isolated online node can be matched: Even if, we just consider a single random edge leaving each of the non-isolated online nodes, the number of pair of such edges that are both incident to the same offline node is $O(|S_2|^2/n)=O(c|S_2|)$, and thus the matching has size at least $|S_2|(1-O(c))$. The second bound follows by observing that for each arriving non-isolated online node, MPD will match it with probability at least $1-m/n=1-c$. 

Similarly to above, there exists $\eps_2>0$, such that if $c>1/\eps_2$, and $n$ is large enough, then $|1-\frac{\E[|M^*|]}{\E[|S_1|]}|<\frac{1}{1000}$. Moreover, we can choose $\eps_2$ such that if $c>1/\eps_2$, and $n$ is large enough, it also holds that $|1-\frac{\E[|M|]}{\E[|S_1|]}|<\frac{1}{1000}$. Indeed, if $d\leq 10$, and $\eps_2$ is small enough, then an $1-10/\eps_2$ fraction of the nodes in $S_1$, will have a degree one neighbor, and therefore be matched, and the case $d\geq 10$ reduces to the case $d=10$.

Combining the above bounds it follows from some calculations that if we choose $\eps_3=min(\eps_1,\eps_2)$, then if $c\notin [\eps_3,1/\eps_3]$, and $n$ is large enough, the competitive ratio of MPD is at least $0.99$. 

We next assume that $c\in [\eps_3,1/\eps_3]$. Using this assumption on $c$, we can then choose $\eps_4>0$ such that if $d<\eps_4$, then $99\%$ of the edges $(u,v)$ of the instance, will satisfy that the vertices $u$ and $v$ have degree $1$. As MPD includes all those edges, and the total number of edges of the graph is an upper bound on $|M^*|$, it follows that in this case, the competitive ratio is at least $0.99$. Finally, it is easy to check by splitting into the cases $m\leq n$ and $m>n$ that we can choose $\eps_5>0$, such that if $d\geq 1/\eps_5$, then $|1-\frac{\E[|M|]}{\min(n,m)}|<\frac{1}{1000}$. As $|M^*|\leq \min(m,n)$, this gives that the competitive ratio in this case is at least $0.99$. Setting $\eps=\min(\eps_3,\eps_4,\eps_5)$, we obtain the desired result. 
\end{proof}

Next, we show how to use the techniques of Appendix~\ref{appendix-asymptotic} to obtain an upper bound on the maximum matching size in the case of Erd\H{o}s-Rényi random bipartite graphs in the case that $d,c=\Theta(1)$.

\begin{lemma}\label{lem:uni-upper-bound}
Let $c,d$ be given with $c,d=\Theta(1)$. Define
\[
A(c,d)=1-e^{-cde^{-d}}, \quad \text{and} \quad B(c,d)=e^{-d}(e^{-dA(c,d)}+d(1-A(c,d)))
\]
For instances $I_{\textbf{p},\textbf{q}}$, where $\textbf{p}=(p,\dots,p)\in [0,1]^n$, and $\textbf{q}=(1,\dots,1)\in [0,1]^m$, the expected fraction of matched offline vertices is at least $(1\pm o(1)) C(c,d)$ as $n\to \infty$, where $C(c,d)=1+cA(c,d)-B(c,d)$.
\end{lemma}
\begin{proof}
Let $N(S)$ denote the set of nodes of $V$ with at least one degree one neighbor in $U$, and let $S$ denote the nodes of $U$ whose neighborhood is fully contained in $N(S)$. The suggestive notation is justified as $N(S)$ is indeed the neighborhood of $S$. As we saw in the Appendix~\ref{appendix-asymptotic}, for $v\in V$,
\[
\Pr[v\in N(S)]=1-\left(1-\frac{d}{n}\left(1-\frac{d}{n}\right)^{m-1}\right)^n=(1\pm o(1))\left(1-e^{-cde^{-d}}\right)=(1\pm o(1))A(c,d),
\]
where we have used the approximation $e^x$ for $1+x$ which is sharp enough to get the bound above, since $d=\Theta(1)$ and $m=\Theta(n)$ with the assumptions on $c$ and $d$. Next for a fixed node $u\in U$, we let $A_k$ denote the event that $u$ has degree $k$. Then for any $\Delta$, 
\[
\Pr[u\in S ]\geq \sum_{k=0}^\Delta \Pr[u\in S \cap A_k].
\]
Assuming that $\Delta=O(1)$ and $2\leq k\leq \Delta$, we can approximate
\[
\Pr[u\in S \cap A_k]=(1\pm o(1))\binom{m}{k}\left(\frac{d}{m}\right)^k\left(1-\frac{d}{m}\right)^{m-k}A(c,d)^k=(1\pm o(1))e^{-d}\frac{(A(c,d)d)^k}{k!}.
\]
Moreover, $\Pr[u\in S \cap A_0]=\Pr[A_0]$ and $\Pr[u\in S \cap A_1]=\Pr[A_1]$, so we can bound 
\[
\Pr[u\in S \cap A_0]=(1\pm o(1))e^{-d} \quad \text{and} \quad \Pr[u\in S \cap A_1]=(1\pm o(1))de^{-d}.
\]
Now for any constant $\Delta$, we can upper bound the expected fraction of matched offline vertices by  
\begin{align*}
&1-\Pr[u\in S]+c\Pr[v\in N(s)]\leq 1-\left( \sum_{k=0}^\Delta \Pr[u\in S \cap A_k]\right)+c\Pr[v\in N(s)] \\
=&1-(1\pm o(1))e^{-d}\left(1+d+\sum_{k=2}^\Delta\frac{(A(c,d)d)^k}{k!}\right)+(1\pm o(1))cA(c,d)
\end{align*}
Since this bound holds for any constant $\Delta$, and the series converges, it also holds in the limit. Thus, we can conclude that as $n\to \infty$, the expected fraction of matched offline nodes is at most
\[
(1\pm o(1))\left(1-e^{-d}(e^{A(c,d)d}+d(1-A(c,d)))+cA(c,d)\right)=(1\pm o(1))(1-B(c,d)+cA(c,d))=(1\pm o(1))C(c,d),
\]
as desired. 
\end{proof}

\paragraph{Remark} When bounding the size of the maximum matching, with Lemma~\ref{lem:uni-upper-bound}, we can switch the roles of $m$ and $n$. Letting $d'=pn$ and $c'=n/m$, it holds that $c'=1/c$ and $d'=d/c$, and we can use Lemma~\ref{lem:uni-upper-bound}, to conclude that the expected fraction of matched \emph{online} nodes is at most $(1\pm o(1))C(1/c,d/c)$ as $n\to \infty$. In particular, we can upper bound the expected fraction of \emph{offline} nodes in a maximum matching by 
$(1\pm o(1))D(c,d)$, where 
\[
D(c,d)=\min(C(c,d),c C(1/c,d/c))
\]

We next show that for Erd\H{o}s-Rényi bipartite random graph, the asymptotic competitive ratio of MPD is at least $0.831$. 
\begin{lemma}
Let $I_{\textbf{p},\textbf{q}}$, where $\textbf{p}=(p,\dots,p)\in [0,1]^n$, and $\textbf{q}=(1,\dots,1)\in [0,1]^m$ be an Erd\H{o}s-Rényi bipartite random graph. The asymptotic competitive ratio of MPD on $I_{\textbf{p},\textbf{q}}$ as $n\to \infty$ is at least $0.831$.
\end{lemma}
\begin{proof}
By Lemma~\ref{lem:assume-constants}, we can assume that $(c,d)\in [\eps,1/\eps]^2$ for some small enough constant, as otherwise, the competitive ratio of MPD is at least 0.99. Then we are in position to apply Theorem~\ref{thm:mpd_uniform} and Lemma~\ref{lem:uni-upper-bound} (in fact the remark following the lemma). Letting $E(c,d)=1+c-\frac{c\ln(e^{d}+e^{d/c}-1)}{d}$, we conclude that the asymptotic competitive ratio of MPD is at most $\frac{E(c,d)}{D(c,d)}$. Using some calculus, it can be checked that in any region $[\eps,1/\eps]^2$, where $\eps$ is sufficiently large, $\frac{E(c,d)}{D(c,d)}$ has a unique minimum attained at $(c_0,d_0)$, where $c_0=1$ and $d_0\approx 2.7997$. Moreover,  $\frac{E(c_0,d_0)}{D(c_0,d_0)}\approx 0.83105\geq 0.831$. 
\end{proof}
\color{black}

\section{Concentration of MinPredictedDegree}
\label{appendix-mpd-concentration}
In this section, we prove that MinPredictedDegree's performance on \GraphModel{} random graphs is concentrated about its expectation.
\begin{theorem}
\label{thm-mpd-concentration}
Let $G$ be a symmetric \GraphModel{} random graph with expected offline degrees $\mathbf{d}$ and let $X$ be the random variable corresponding to the size of the matching returned by MPD.
Then,
\begin{equation}
\label{eq-mpd-concentration}
    \Pr(|X - \E[X]| \geq 2 \sqrt{m} \log m) \leq \frac{2}{m}.
\end{equation}
\end{theorem}

\begin{proof}
Let $H_j$ represent the state of MinPredictedDegree after it has processed the $j$th online node, and let $Y_j = \E[X | H_j]$ be the expectation of the size of the returned matching conditioned on the history of the algorithm up to time $j$.
Then $\{Y_j\}_{j=0}^m$ form a Doob martingale.
We will proceed by bounding $|Y_j - Y_{j-1}|$.

If an offline node $i$ was matched with online node $j$, then the conditional expectation of the final matching size increases by $1 - \Pr(i \text{ matched} | H_{j-1})$.
For each unmatched offline node that is not matched with online node $j$, the conditional expectation decreases by the sum over all such nodes $i'$ of $\Pr(i' \text{ matched with } j |H_{j-1})$.
As both the increment and decrement are bounded in magnitude by $1$, the martingale has bounded differences $|Y_j - Y_{j-1}| \leq 1$.

Applying the standard Azuma's inequality bounds, we get the concentration result:
\begin{align*}
    \Pr(|Y_m - Y_0| \geq 2 \sqrt{m} \log m) \leq \frac{2}{m}.
\end{align*}
As $Y_0 = \E[X]$ and $Y_m = X$, this completes the proof.
\end{proof}

\section{Concentration of the Upper Bound on Maximum Matching Size}
\label{appendix-offline-concentration}
In this section, we show that our upper bound on the size of a maximum matching in \GraphModel{} random graphs with power law distributed degrees is concentrated about its expectation.

\begin{theorem}
\label{thm-offline-concentration}
Let G be a \GraphModel{} random graph with $n=m$ and with expected offline degrees following a power law distribution with exponent $\alpha > 3$, and let $X$ be the random variable corresponding to the difference $|S^*| - |N(S^*)|$ where $S^*$ and $N(S^*)$ are the subsets of the nodes in $G$ described in Section~\ref{appendix-offline-expectation}.
Then, there exists some constant $C$ s.t. 
\begin{equation}
\label{eq-offline-concentration-thm}
    \Pr(|X - \E[X]| \geq C \sqrt{n} \log n) \leq \frac{1}{n}.
\end{equation}
\end{theorem}

\begin{proof}[Proof of Theorem~\ref{thm-offline-concentration}]
Let $H_t$ represent the $t$ offline nodes with the smallest degrees (ties broken arbitrarily) as well as their incident edges and let $Y_t =\E[\left(|S^*| - |N(S^*)|\right) |  H_t]$ be the expected difference in the sizes of $S^*$ and $N(S^*)$ given knowledge of $H_t$.
Note that here we are using true degrees and \emph{not} expected degrees.
Then, $\{Y_t\}_{t=0}^n$ form a Doob martingale.

Let $u_t$ be the offline node with the $t$th smallest degree and let $deg(u_t)$ be the its degree.
We will proceed by cases to show that the martingale has bounded differences.
\begin{enumerate}[leftmargin=*]
    \item Assume $deg(u_t) > 1$.
    Then, from $H_{t-1}$ we know $N(S^*)$ as all degree one nodes have already been seen.
    Therefore, the contribution of $u_t$ to the difference $\left(|S^*| - |N(S^*)|\right)$ is independent of any subsequent offline nodes.
    Specifically, if $u_t \in S^*$, $Y_t - Y_{t-1} = 1 - \Pr(u_t \in S^* | H_{t-1})$, and if $u_t \notin S^*$, $Y_t - Y_{t-1} = -\Pr(u_t \in S^* | H_{t-1})$ where
    \begin{equation}
    \label{eq-offline-concentration1}
    \begin{split}
        \Pr(&u_t \in S^* | H_{t-1}) = 
        \E_{deg(u_t) | H_{t-1}}\left[\prod_{k=0}^{deg(u_t)-1} \left(\frac{N(S^*)-k}{m}\right)\right].
    \end{split}
    \end{equation}
    In any case, $|Y_t - Y_{t-1}| \leq 1$.
    \item Assume $deg(u_t) \leq 1$.
    In this case, we have to deal with the fact that $u_t$ can affect $N(S^*)$ as well as $S^*$.
    Part of the difference $Y_t - Y_{t-1}$ is due to the inclusion of $u_t$ in $S^*$ and the subsequent possibility that $u_t$ contributes a node to $N(S^*)$ if $deg(u_t)=1$ and its neighbor is not already in $N(S^*)$.
    This part of the difference is bounded in magnitude by one as these events change the difference $S^* - N(S^*)$ by at most one.
    
    The other part of the difference $Y_t - Y_{t-1}$ is due to whether $u_t$ increments the size of $N(S^*)$ via its neighbor, affecting the probabilities $\Pr(u_{t'} \in S^*)$ for $t'$ where $deg(u_{t'}) > 1$ as in Equation~\ref{eq-offline-concentration1}.
    As the expected size of $N(S^*)$ can change by at most one, the change in $\Pr(u_{t'} \in S^*)$ for each $t'$ where $deg(u_{t'}) > 1$ is at most
    \begin{equation}
    \label{eq-offline-concentration2}
        \prod_{k=0}^{deg(u_{t'})-1} \left(\frac{t'-k}{m}\right) - \prod_{k=0}^{deg(u_{t'})-1} \left(\frac{t'-k-1}{m}\right).
    \end{equation}
    The factors of $t'$ in the numerators come from the fact that $N(S^*) \leq t'$ if $deg(u_{t'}) > 1$.
    As both parts of the difference contain many of the same terms, we can simplify Expression~\ref{eq-offline-concentration2} as
    \begin{equation}
    \label{eq-offline-concentration3}
        \frac{deg(u_{t'})}{m} \prod_{k=1}^{deg(u_{t'})-1} \left(\frac{t'-k}{m}\right) \leq \frac{deg(u_{t'})}{m}.
    \end{equation}
    As we assume that the expected degrees are distributed according to a power law distribution with exponent $\alpha > 3$, the expectation and variance of the degree of a given node $u$ will be constant. Thus, with high probability, the sum over all offline degrees $\sum_u deg(u) = O(m)$. 
    The contribution to $Y_t -Y_{t-1}$ by Expression~\ref{eq-offline-concentration3} is thus bounded by $\sum_u \frac{deg(u)}{m} = O(1)$.    
    Overall, with high probability, $|Y_t - Y_{t-1}| = O(1)$.
\end{enumerate}
As in both cases, the martingale has constant bounded differences (with high probability), Azuma's inequality directly gives us the theorem.
\end{proof}

\section{Additional Experiments}
\label{appendix-experiments}

\begin{figure}[ht]
\begin{center}
\centerline{\includegraphics[width=0.5\columnwidth]{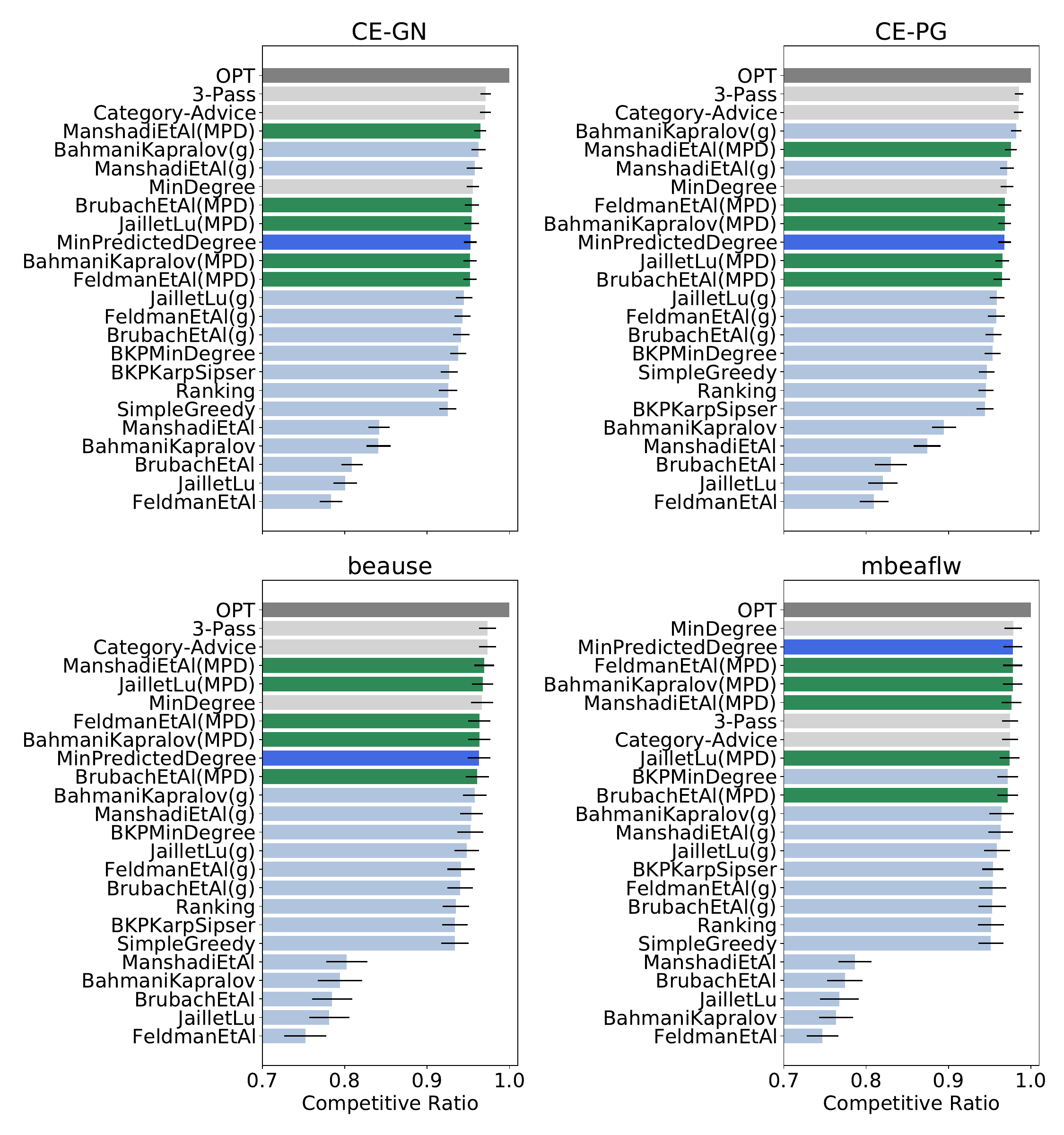}}
\caption{Additional comparison of empirical competitive ratios on Real World graphs.
Algorithms depicted in gray are \emph{not} online algorithms (they use extra information or multiple passes). Algorithms in green are augmented with MPD.}
\label{fig-realworld-appendix}
\end{center}
\end{figure}

Figure~\ref{fig-realworld-appendix} shows additional experiments on Real World graphs from the known i.i.d.\ model (based on the methodology of~\cite{borodin2020experimental}). Overall, the results are very similar to those in Section~\ref{sec-experiments}, MinPredictedDegree does very well compared to the other online baselines (depicted in light blue) despite making relatively little use of the type graph information. Additionally, augmenting the known i.i.d.\ baselines with MinPredictedDegree (e.g.\ using the MinPredictedDegree rule when the base algorithm does not match the current online node even though it has unmatched neighbors) often improves the performance over the baseline algorithm and the greedy augmentation.

\begin{figure}[ht]
\begin{center}
\begin{subfigure}[b]{0.5\columnwidth}
    \centerline{\includegraphics[width=\textwidth]{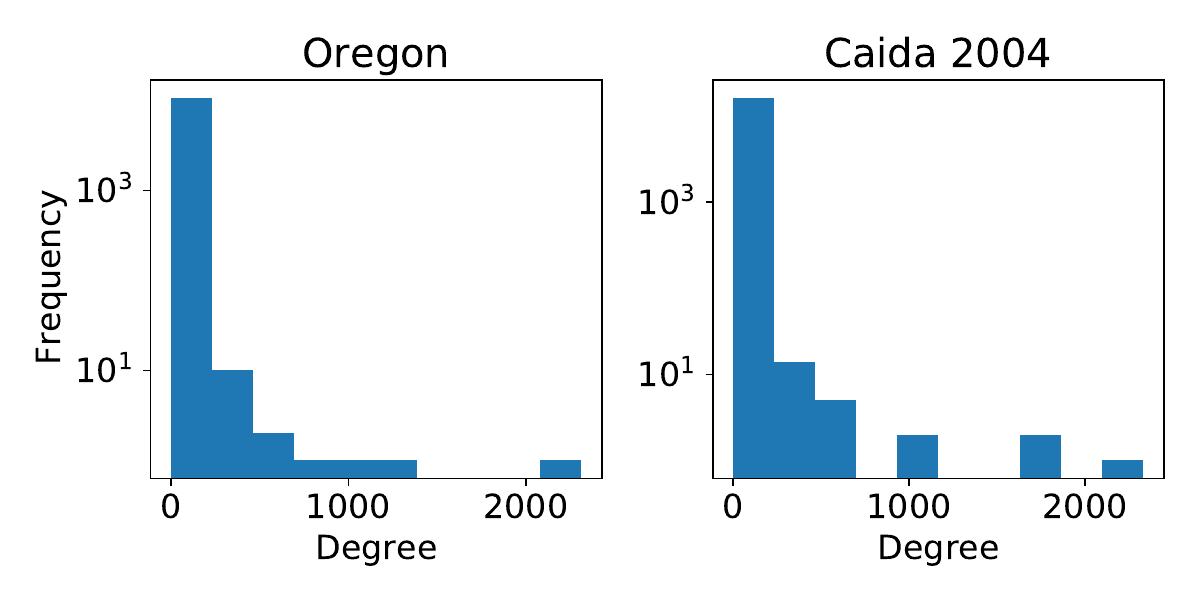}}
\end{subfigure}

\begin{subfigure}[b]{0.35\textwidth}
    \centerline{\includegraphics[width=\textwidth]{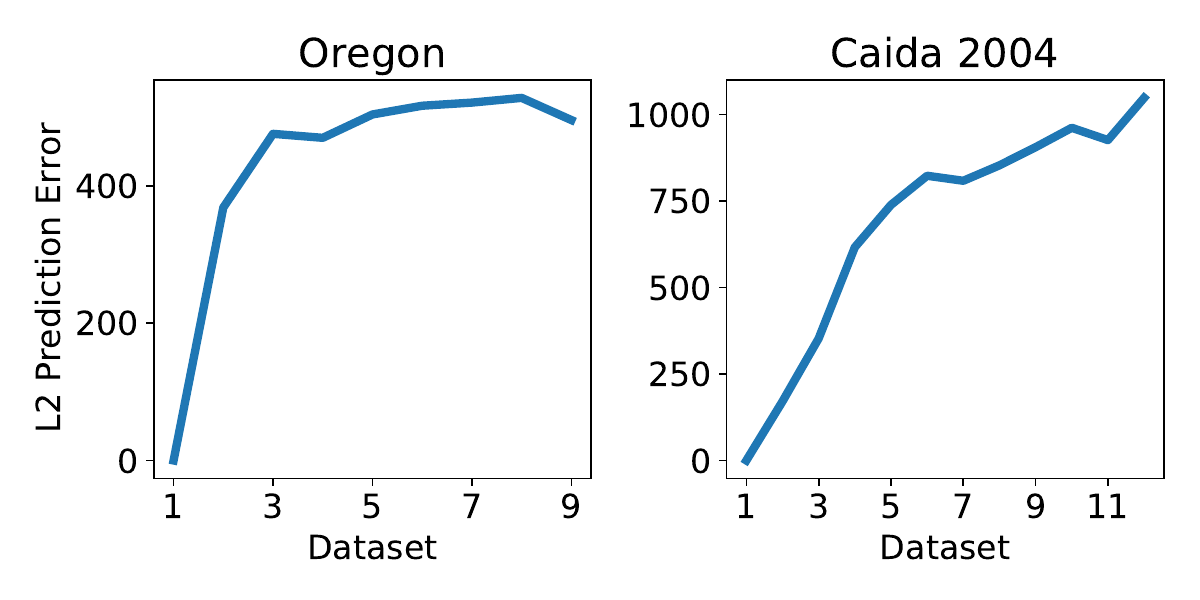}}
\end{subfigure}
\caption{Degree distribution (left) and $\ell_2$ prediction error over time (right) for the Oregon and Caida 2004 datasets.}
\label{fig-oracle-analysis}
\end{center}
\end{figure}

Figure~\ref{fig-oracle-analysis} shows the degree distribution of the Oregon and Caida 2004 datasets as well as the $\ell_2$ prediction error (square root of the sum of the squared error of the degree prediction for each offline node in the current graph) over time of using the first days degrees as a prediction for future degrees. As the prediction quality degrades, the performance of MinPredictedDegree slowly declines.

\end{document}